%% file: Offline_Algorithms.tex
\RequirePackage{etoolbox}
\newtoggle{conference}

 \togglefalse{conference} % For the arXive version

\newcommand{\inConference}[1]{\iftoggle{conference}{#1}{}} % For text that should appear in the conference format only
\newcommand{\inArXiv}[1]{\iftoggle{conference}{}{#1}}  % For text that should appear in the Arxiv format only
\documentclass[11pt]{article}
\usepackage{graphicx}
\usepackage{latexsym}
\usepackage{amssymb}
\usepackage{amsmath}
\usepackage{amsthm}
\usepackage{forloop}
\usepackage[paper=letterpaper,margin=1in]{geometry}
\usepackage[ruled,vlined,linesnumbered]{algorithm2e}
\usepackage{paralist}
\usepackage{caption}
\usepackage{hyperref}

\DeclareCaptionType{mechanism}[Mechanism][List of Mechanisms]
\newtheorem{theorem}{Theorem}[section]
\newtheorem{lemma}[theorem]{Lemma}

\newtheorem{corollary}[theorem]{Corollary}

\newtheorem{observation}[theorem]{Observation}

\makeatletter
\newtheorem*{rep@theorem}{\rep@title}
\newcommand{\newreptheorem}[2]{%
\newenvironment{rep#1}[1]{%
 \def\rep@title{#2 \ref{##1}}%
 \begin{rep@theorem}}%
 {\end{rep@theorem}}}
\makeatother

\newreptheorem{theorem}{Theorem}
\newreptheorem{reduction}{Reduction}
\newreptheorem{lemma}{Lemma}
\newreptheorem{observation}{Observation}

\newcommand{\defcal}[1]{\expandafter\newcommand\csname c#1\endcsname{{\mathcal{#1}}}}
\newcommand{\defbb}[1]{\expandafter\newcommand\csname b#1\endcsname{{\mathbb{#1}}}}
\newcounter{calBbCounter}
\forLoop{1}{26}{calBbCounter}{
    \edef\letter{\Alph{calBbCounter}}
    \expandafter\defcal\letter
		\expandafter\defbb\letter
}

\newcommand{\GfT}{\mathtt{GfT}}
\newcommand{\Val}{\mathtt{Val}}
\newcommand{\TPM}{{\texttt{TPM}}}
\newcommand{\PRM}{{\texttt{PRM}}}
\newcommand{\ie}{{\it i.e.}}

\makeatletter \newcommand{\hypertargettop}[1]{\Hy@raisedlink{\hypertarget{#1}{}}}
\makeatother

\title{\textbf{Double-Sided Markets with Strategic Multi-dimensional Players}\footnote{This work is supported by the Horizon 2020 funded project TYPES (Project number: 653449. Call Identifier H2020-DS-2014-1). 
\inConference{We are submitting this paper for confidential review to be considered for publication in STOC on June 19--22, 2017.}\inArXiv{We are submitting this paper for confidential review to be considered for publication in 2017.}}}

\author{
 Moran Feldman\thanks{Department of Mathematics and Computer Science, the Open University of Israel. E-mail: \href{mailto:moranfe@openu.ac.il}{moranfe@openu.ac.il}.}
 \and
 Rica Gonen\thanks{Department of Management and Economics, the Open University of Israel. E-mail: \href{mailto:gonenr@openu.ac.il}{gonenr@openu.ac.il}}
}

\begin{document}

\input{MovableTexts}
\maketitle
\input{Abstract}

\pagenumbering{Alph}
\thispagestyle{empty}
\newpage
\setcounter{page}{1}
\pagenumbering{arabic}

\input{SecIntroduction}
\input{SecModel}
\input{SecCanonical}
\input{SecDeterministic}
\input{SecRandomized}
\input{SecConclusions}

\bibliographystyle{plain}
\bibliography{../TwoSidedTrade}

\inConference{
\appendix
\input{AppMissingProofs}

\input{AppDeterministic}
\input{AppRandomized}
}

\end{document}

%% file: MovableTexts.tex
\newcommand{\proofSlotsConsecutive}{
\begin{proof}
Let $b_1$ and $b_2$ be two slots of one advertiser, and let $b$ be an arbitrary slot of another advertiser. It is enough to prove that $v(b_1) < v(b)$ implies $v(b_2) < v(b)$. The inequality $v(b_1) < v(b)$ can happen in two cases. If $v(b_1)$ is smaller than $v(b)$ as numbers, then we get $v(b_2) < v(b)$ since $v(b_1)$ and $v(b_2)$ are equal as numbers. Otherwise, if $v(b_1)$ is equal to $v(b)$ as numbers then the inequality $v(b_1) < v(b)$ implies that the advertiser of $b_1$ (and $b_2$) appears later in $\sigma$ than the advertiser of $b$, and thus, we also have $v(b_2) < v(b)$.
\end{proof}
} % End of \proofSlotsConsecutive

\newcommand{\proofCanonicalOptimal}{
\begin{proof}
For every $0 \leq i \leq \min\{|B'|, |P'|\}$, let $S^i_c(P', B')$ be the subset of $S_c(P', B')$ in which we keep only the assignments of the first $i$ users. More formally,
\[
	S^i_c(P', B')
	=
	\{(p_j, b_j) \in S_c(P', B') \mid 1 \leq j \leq i\}
	\enspace.
\]

Let $i^*$ be the largest $i$ such that $S^i_c(P', B')$ is contained in some assignment $S^* \subseteq P' \times B'$ maximizing $\GfT(S^*)$ among all the possible assignments of users of $P'$ to slots of $B'$. Such an $i^*$ clearly exists since $S^0_c(P', B') = \varnothing$ is a subset of every feasible assignment. There are two cases to consider. The first case is when $i^* = \min\{|B'|, |P'|\}$. In this case the canonical assignment $S_c(P', B')$ is a subset of the optimal assignment $S^*$. Let us consider an arbitrary ordered pair $(p_j, b_k) \in S^* \setminus S_c(P', B')$, and let us assume, without loss of generality, that $j \leq k$. Since $S^*$ is a legal assignment and $(p_j, b_k)$ belongs to $S^*$ together with all the ordered pairs of $S_c(P', B')$, the ordered pair $(p_j, b_j)$ cannot belong to $S_c(P', B')$. By definition, this implies $c(p_j) > v(b_j)$, and thus, $c(p_j) \geq v(b_j)$ even when we compare them as numbers. Thus,
\[
	v(b_k) - c(p_j)
	\leq
	v(b_j) - c(p_j)
	\leq
	0
	\enspace.
\]
Summing the above inequality over all ordered pairs $(p_j, b_k) \in S^* \setminus S_c(P', B')$, we get:
\[
	\GfT(S_c(P', B'))
	=
	\GfT(S^*) - \sum_{(p_j, b_k) \in S^* \setminus S_c(P', B')} \mspace{-36mu} [v(b_k) - c(p_j)]
	\geq
	\GfT(S^*)
	\enspace,
\]
which completes the proof of the lemma for the case $i^* = \min\{|B'|, |P'|\}$. The second case we need to consider is the case $0 \leq i^* < \min\{|B'|, |P'|\}$. In the rest of the proof we show that this case can never happen, which implies the lemma.

Assume towards a contradiction $0 \leq i^* < \min\{|B'|, |P'|\}$. Our objective is to show that there exists an assignment $S'$ obeying $S^{i^* + 1}_c(P', B') \subseteq S'$ and $\GfT(S') \geq \GfT(S^*)$, which is a contradiction to the choice of $i^*$. By the choice of $i^*$ we have $S^{i^*}_c(P', B') \subseteq S^*$, but $S^{i^* + 1}_c(P', B') \not \subseteq S^*$, which can only happen when $(p_{i^* + 1}, b_{i^* + 1}) \in S^{i^* + 1}_c(P', B') \setminus S^*$. There are now three cases that need to be considered:
\begin{itemize}
	\item If $(p_{i^* + 1}, b_{i^* + 1})$ can be added to $S^*$ without violating the feasibility, then we choose $S' = S^* \cup \{(p_{i^* + 1}, b_{i^* + 1})\}$. Clearly, $S^{i^* + 1}_c(P', B')$ is a subset of $S'$ since $S^{i^* + 1}_c(P', B') = S^{i^*}_c(P', B') \cup \{(p_{i^* + 1}, b_{i^* + 1})\}$ and $S^{i^*}_c(P', B')$ is a subset of $S^*$. Additionally, since $(p_{i^* + 1}, b_{i^* + 1})$ belongs to the canonical assignment, we know that $v(b_{i^* + 1}) > c(p_{i^* + 1})$, and thus, $v(b_{i^* + 1}) \geq c(p_{i^* + 1})$ also when we compare them as numbers. Thus,
	\[
		\GfT(S')
		=
		\GfT(S^*) + [v(b_{i^* + 1}) - c(p_{i^* + 1})]
		\geq
		\GfT(S^*)
		\enspace.
	\]
	
	\item When the previous case does not hold, there must be in $S^*$ either an ordered pair containing $p_{i^* + 1}$ or an ordered pair containing $b_{i^* + 1}$. In the current case we assume that $S^*$ contains an ordered pair containing $p_{i^* + 1}$, but does not contain an ordered pair containing $b_{i^* + 1}$. The case that $S^*$ contains both an ordered pair containing $p_{i^* + 1}$ and an ordered pair containing $b_{i^* + 1}$ is the next case we consider. Finally, the case that $S^*$ contains an ordered pair containing $b_{i^* + 1}$, but does not contain an ordered pair containing $p_{i^* + 1}$ is analogous to the current case, and is, thus, omitted.
	
	Let $(p_{i^* + 1}, b)$ be the ordered pair of $S^*$ that contains $p_{i^* + 1}$. Since $(p_{i^* + 1}, b_{i^* + 1})$ appears in the canonical solution, we must have $c(p_{i^* + 1}) < v(b_{i^* + 1})$, and thus, $c(p_i) < v(b_i)$ for every $1 \leq i \leq i^*$. Hence, the assignment $S^{i^*}_c(P', B')$, which is contained in $S^*$, is equal to $\{(p_i, b_i) \mid 1 \leq i \leq i^*\}$. This means that the slot $b$ of the pair $(p_{i^* + 1}, b) \in S^*$ cannot be one of the first $i^*$ slots, and thus, obeys $v(b) \leq v(b_{i^* + 1})$. We now consider the assignment $S' = S^* \cup \{(p_{i^* + 1}, b_{i^* + 1})\} \setminus \{(p_{i^* + 1}, b)\}$. Clearly this assignment contains $S^{i^* + 1}_c(P', B')$ since $S^{i^* + 1}_c(P', B') = \{(p_i, b_i) \mid 1 \leq i \leq i^* + 1\}$ and $S^*$ contains the assignment $S^{i^*}_c(P', B') = \{(p_i, b_i) \mid 1 \leq i \leq i^*\}$. Additionally,
	\begin{align*}
		\GfT(S')
		={} &
		\GfT(S^*) + [v(b_{i^* + 1}) - c(p_{i^* + 1})] - [v(b) - c(p_{i^* + 1})]\\
		={} &
		\GfT(S^*) + [v(b_{i^* + 1}) - v(b)]
		\geq
		\GfT(S^*)
		\enspace.
	\end{align*}
	
	\item The last case we need to consider is when $S^*$ contains both an ordered pair $(p_{i^* + 1}, b)$ containing user $p_{i^* + 1}$ and an ordered pair $(p, b_{i^* + 1})$ containing slot $b_{i^* + 1}$. In this case we choose $S' = S^* \cup \{(p_{i^* + 1}, b_{i^* + 1}), (p, b)\} \setminus \{(p_{i^* + 1}, b), (p, b_{i^* + 1})\}$. Since $S^{i^*}_c(P', B') \subseteq S^*$ and $S^{i^*}_c(P', B')$ cannot contain (by definition) either the ordered pair $(p, b_{i^* + 1})$ or the ordered pair $(p_{i^* + 1}, b)$, we get that $S'$ contains $S^{i^* + 1}_c(P', B')$. Additionally,
	\begin{align*}
		\GfT(S')
		={} &
		\GfT(S^*) + [v(b_{i^* + 1}) - c(p_{i^* + 1})] + [v(b) - c(p)] \\ &- [v(b) - c(p_{i^* + 1})] - [v(b_{i^* + 1}) - c(p)]
		=
		\GfT(S^*)
		\enspace.
		\qedhere
	\end{align*}
\end{itemize}
\end{proof}
} %End of \proofCanonicalOptimal

%% file: Abstract.tex
\begin{abstract}
We consider mechanisms for markets that are double-sided and have players with multi-dimensional strategic spaces on at least one side. The players of the market are strategic, and act to optimize their own utilities. The mechanism designer, on the other hand, aims to optimize a social goal, \ie, the gain from trade. 
%As the mediators control the information flow from their players to the mechanism, the mechanism faces strategic behavior not only from the players but also from mediators: a mediator acts strategically to maximize utility on behalf of the players he represents.
We focus on one example of this setting which is motivated by the foreseeable future form of online advertising.

Online advertising currently supports some of the most important Internet services, including: search, social media and user generated content sites. To overcome privacy concerns, it has been suggested to introduce user information markets through information brokers into the online advertising ecosystem. Such markets give users control over which data get shared in the online advertising exchange. We describe a model for the above foreseeable future form of online advertising, and design two mechanisms for the exchange of this model: a deterministic mechanism which is related to the vast literature on mechanism design through trade reduction and allows players with a multi-dimensional strategic space, and a randomized mechanism which can handle a more general version of the model.

\end{abstract}

%% file: SecIntroduction.tex
\section{Introduction}

Billions of transactions are carried out via exchanges at every given day, and the number of transactions and exchanges continues to grow as the need for competitiveness promotes adoption. The design of one-sided incentive compatible (truthful bidding) mechanisms for exchanges is relatively well understood. However, incentive compatible multi-sided mechanisms present a significantly more challenging problem as they introduce more sophisticated requirements such as budget balance. %Moreover, it is often assumed in mechanism design that players interact directly with the mechanism computing the outcome. Yet, in many complex real world settings interactions pass through intermediaries acting on the players' behalf. 

%One example for a setting demonstrating the above issues is an exchange accessed through brokers. There are sellers, each having a cost for selling a unit of a good (for example, a stock, an item or a piece of information). The sellers sell their units through brokers, where each broker may represent a different number of sellers. On the other side of the market there are buyers, each having a value for buying a unit of a good and a capacity determining the maximum number of units they are interested in buying. The parties engaged directly in the exchange (buyers and brokers) aim to maximize their personal utilities,\footnote{The utility of a broker is the difference between the payment made to him and the total costs of his sellers (the amount he pays his sellers).  The utility of a buyer is the difference between her value for the outcome and the amount she pays.} The global goal is to maximize the \emph{gain from trade}---the difference between the total value of the sold goods for the buyers and the total costs of these goods for the sellers. Towards this goal, the mechanism can make decisions based, exclusively, on the reports of the brokers and buyers. In particular, the mechanism has no direct interaction with the sellers, and the brokers are free to make arbitrary manipulations to the information they transfer to the mechanism about their sellers.

More specifically, we are interested in designing exchanges (mechanisms) for multi-sided markets with strategic players. The players of the market are strategic, and act to optimize their own utilities. The mechanism designer, on the other hand, aims to optimize a social goal, \ie, the gain from trade (the difference between the total value of the sold goods for the buyers and the total costs of these goods for the sellers). The design of such mechanisms raises a few interesting questions. Can the mechanism maintain simultaneously different desirable economic properties such as: individual rationality (IR)---participants do not lose by participation, incentive compatibility (IC)---players are incentivized to report their true information to the mechanism and budget balance (BB)---the mechanism does not end up with a loss. Moreover, can the mechanism maintain these properties while suffering only a bounded loss compared to the optimal gain from trade? Finally, can this be done when \textbf{all} the players have a \textbf{multi-dimensional strategic space}?\footnote{We often refer to players with a multi-dimensional strategic space as multi-dimensional players.}

 %Each mediator does not have private information of his own other than the number of players associated with him. Opposing players not interacting through a mediator may also have a multi-dimensional strategic space.
%As the mediators control the information flow from their sellers to the mechanism, the mechanism faces strategic behavior not only from the buyers but also from mediators: a mediator acts strategically to maximize his own utility. In this paper we aim to design mechanisms for double-sided exchanges with multi-dimensional players\footnote{Multi-dimensional players are players with a multi-dimensional strategic space.} having similar properties to the ones guaranteed by known mechanisms for simpler exchanges (such as a double-sided exchange with single-dimensional players). 

The above questions can be studied in the context of many multi-sided markets. We focus on one such market, and leave the consideration of other multi-sided markets for future work. The market we consider is motivated by online advertising in its foreseeable future form. Online advertising currently supports some of the most important Internet services, including: search, social media and user generated content sites. However, the amount of information that companies collect about users increasingly creates privacy concerns in society as a whole, and even more so in the European society. In recent years EU regulators have actively been looking for solutions to guarantee that users' privacy is preserved. % In consideration is a system that gives end users control over the amount of information they are willing to share and who they will share it with.
In particular, the EU regulators have been looking for tools that enable the end user to configure their privacy settings so that only information allowed by the end-user is collected by online advertising platforms.  

The market we study is induced by a new solution we suggest for the above privacy issue. In this solution mediators serve as the interface between end-users and the other players in the online advertising market. Each user informs his mediator of the attributes she is willing to reveal, and her cost, \ie, the compensation she requires for every ad she views. The mediator then tries to ``sell'' the user on the advertising market based \emph{only} on the attributes she agreed to reveal, and, if successful, pays her the appropriate cost out of the amount he got from the sell.

As revealing more attributes is likely to result in a more profitable sale, our solution provides incentives for users to share their information with the advertising market while allowing users to retain control of the amount of information they would like to share. %\footnote{Our solution is privacy preserving in the sense that it targets the user based solely on the user's attributes that the user agreed to share with the market.}
Notice that the fact that our solution motivates users to participate in the advertising market, and even to provide more precise information for targeting campaigns, means that our solution improves the efficiency of the advertising system and the digital economy as a whole (in addition to answering the privacy concerns discussed above). This is in sharp contrast to other natural approaches for dealing with privacy issues, such as cryptography based approaches, which reduce the amount of information available to the advertisers but give them nothing in return.

%As a result, there has been growing adoption of tools for blocking any transfer of information from end users towards the online advertising ecosystem. A massive adoption of these tools by end users may cause disruptions in the digital economy. To overcome the above, user information markets through information brokers were suggested. Such markets give users control over which data get shared in the online advertising exchange. 

%The model we use for the above solution consists of a set $P$ of users, a set $M$ of mediators and a set $A$ of advertisers. The users are partitioned into $|M|$ disjoint sets, and each set of users is represented by a unique mediator. Each user has a non-negative cost, which she has to be paid if assigned to an advertiser. The mediators have no cost of their own, however, each of them has to pay his users their cost if they are assigned to advertisers. Thus, the utility of a mediator is the amount paid to him minus the total cost of his users that are assigned.
%Finally, each advertiser has a positive capacity determining the number of users she is interested in targeting, and she gains a non-negative value from every one of the users assigned to her (as long as her capacity is not exhausted). Thus, the advertiser's utility is her value multiplied by the number of users assigned to her (as long as this number is less than or equal her capacity) minus her total payment.

The advertising market induced by the above solution has mediators on one side, and advertisers on the other side. Each mediator has a set of users associated with him, and he is trying to assign these users to advertisers using the market. Each one of the users has a non-negative cost which she has to be paid if she is assigned to an advertiser. The mediators themselves have no cost of their own, however, each of them has to pay his users their cost if they are assigned to advertisers. Thus, the utility of a mediator is the amount paid to him minus the total cost of his users that are assigned.
Finally, each advertiser has a positive capacity determining the number of users she is interested in targeting, and she gains a non-negative value from every one of the users assigned to her (as long as her capacity is not exhausted). Thus, the advertiser's utility is her value multiplied by the number of users assigned to her (as long as this number does not exceed her capacity) minus her total payment.

A mechanism for the above market knows the mediators and the advertisers, but has no knowledge about their parameters or about the users. The objective of the mechanism is to find an assignment of users to advertisers that maximizes the gain from trade. In addition, the mechanism also decides how much to charge (pay) each advertiser (mediator). In order to achieve these goals, the mechanism receives reports from the advertisers and mediators. Each advertiser reports her capacity and value, and each mediator reports the number of his users and their costs. The mediators and advertisers are strategic, and thus, free to send incorrect reports. In other words, an advertiser may report incorrect capacity and value, and a mediator may report any subset of his users and associate an arbitrary cost with each user. We say that an advertiser is \emph{truthful} if she reports correctly her capacity and value. A mediator is considered \emph{truthful} if he reports to the mechanism his true number of users and the true costs of these users. Notice that we assume that the costs of the users are known to their corresponding mediators, \ie, the users are non-strategic. This assumption is reasonable given the high speed of the online advertising market, especially compared to the speed at which a private user can change her contract with her mediator. %Users are not strategic and are interested in getting paid their cost if assigned. 

To better understand the design challenge raised by this market, we observe that even if our setting is reduced to a single buyer-single seller exchange, still it is well known from \cite{MS83} that maximizing gain from trade while maintaining individual rationality and incentive compatibility perforce to run into deficit (is not budget balanced). A well known circumvention of \cite{MS83}'s impossibility is \cite{M92}'s trade reduction for the simple setting of double sided auctions. In \cite{M92}'s setting trade is conducted between multiple strategic sellers offering identical goods to multiple strategic buyers, where each seller is selling a single good and each buyer is interested in buying a single good. The result of \cite{M92} relaxes the requirement for optimal trade by means of a \emph{trade reduction}. The trade reduction leads to an individually rational, incentive compatible and budget balance mechanism. Following \cite{M92}'s work several other mechanisms were designed using the technique of trade reduction. However, all the trade reduction mechanisms suggested in the literature to date allow only players with single dimensional strategic spaces.% These mechanisms, and their corresponding settings, where unified by \cite{GGP07} into a single trade reduction procedure that performs at least as well as the original mechanisms. 

%In our setting, if we require mediators as well as advertisers to have equal capacities, \ie, all mediators will have the same number of users and all advertisers will have the same capacity\footnote{The number of users of a mediator may differ, however, from the advertisers' capacity.} then we could apply \cite{GGP07}'s procedure and achieve all the above economic properties (IR, IC and BB) together with a non-trivial competitive ratio. Nevertheless, if we slightly extend the above limited setting by allowing each advertiser to have a different capacity (but still require all mediators to have an equal number of users), the trade reduction procedure of \cite{GGP07} can no longer guarantee both the three economic properties and a non-zero gain from trade.\footnote{A detailed comparison of our setting and results to that of \cite{GGP07} and a more detailed explanation of the above observation can be found in Section~\ref{subsec:related-work}.}% Similarly, if we slightly extend the above limited setting such that advertisers have equal capacities however mediators have different sizes of users sets, again \cite{GGP07} trade reduction procedure also cease to maintain the three economic properties together with a non trivial competitive ratio.

\subsection{Our Contribution} \label{ssc:contribution}
%Given the impossibility to achieve IC, IR, and BB multi-sided mechanisms with optimal gain from trade and
Given that existing trade reduction solutions do not apply in our setting, we describe new double-sided mechanisms able to handle mediators and advertisers with multi-dimensional strategic spaces. Our mechanisms guarantee desirable economic properties, and at the same time yield a gain from trade approximating the optimal gain from trade. 
%Similarly to \cite{BFT16}, if, under some mechanism, being truthful is a dominant strategy\footnote{Here and throughout the paper, a reference to domination of strategies should be understood as a reference to weak domination. We never refer to strong domination.} for each user given that her mediator is truthful (regardless of any parameter of the model, like the number of mediators, and regardless of other players' strategies), then we say that the mechanism is \emph{user-side incentive compatible} (user-side IC). Additionally, if the mechanism is IR for each user given that her mediator is truthful than we say that the mechanism is \emph{user-side individually rational} (user-side IR). Similar terminology can be defined also for describing mediator incentives. Namely, if being truthful is a dominant strategy for each mediator given that all his users are truthful (again, regardless of any parameters and regardless of other players' strategies), then the mechanism is \emph{mediator-side incentive compatible} (mediator-side IC); and if the mechanism is IR for each mediator given that all his users are truthful than it is \emph{mediator-side individually rational} (mediator-side IR). 
If being truthful is a dominant strategy\footnote{Here and throughout the paper, a reference to domination of strategies should be understood as a reference to weak domination. We never refer to strong domination.} for each advertiser and each mediator (regardless of other players' strategies), then the mechanism is \emph{incentive compatible} (IC); if no advertiser and no mediator have a negative utility by participating truthfully in the mechanism than it is \emph{individually rational} (IR). Our objective is to construct mechanisms that are IC and IR. % We construct both deterministic and randomized multi-sided mechanisms, and prove that they perform favorably.

We first study a special case of our setting where the advertisers' capacities are publicly known (however, these capacities need not be all equal). The set of users of each mediator, on the other hand, remains unknown to the mechanism (\ie, the mechanism only learns about it through the mediator's report). For this case we present a deterministic mechanism we term ``Price by Removal Mechanism'' (\PRM) that works as follows: for every mediator find a threshold cost, and remove users of the mediator whose cost is above this threshold. Add a dummy advertiser with value that is the maximum threshold cost computed for the mediators and a capacity that is equal to the total number of users remaining. Assign the non-removed users to the advertisers using a VCG auction~\cite{V61,C71,G73} in which the users are the goods and the bidders are the advertisers.
Price the mediators according to their threshold cost, and price the advertisers according to the prices of the VCG auction describe above.

%We note that every mediator's critical cost is computed with an offset of four times a maximum capacity bound to eliminate potential advertisers' manipulations regarding their value and as a result regarding their position in the allocation. Such manipulations elimination are: unallocated advertisers remain unallocated when reducing their value and allocated advertisers do not change the size of the optimal trade when increasing their value. 
%****Rica: why did you remove this part?*****

The method used to calculate the threshold costs of the above mechanism induce its properties. We prove that, for appropriately chosen threshold costs, the above mechanism is IC, IR, BB and provides a non-trivial approximation for the optimal gain from trade. More formally, if $\tau$ is the size of the optimal trade, and $\gamma$ is an upper bound, known to the mechanism, on the maximum capacity of any player (mediator or advertiser), then:

\begin{theorem} \label{thm:PRM}
{\PRM} is BB, IR, IC and $\left(1 - \frac{5\gamma}{\tau}\right)$-competitive.
\end{theorem}

{\PRM} generalizes the trade reduction ideas used so far in the literature for single dimensional strategic players, but is much more involved. Intuitively, {\PRM} differs from previous ideas by the following observation. A trading set is the smallest set of players that is required for trade to occur. In the existing literature for single dimensional strategic players a trade reduction mechanism makes a binary decision regarding every trading set of the optimal trade, \ie, either the trading set is removed as a whole, or it is kept. On the other hand, dealing with multi-dimensional players requires {\PRM} to remove only parts of some trading sets, and thus, requires it to make non-binary decisions.

Our deterministic mechanism {\PRM} handles one type of multi-dimensional players (the mediators) and one type of single dimensional strategic space players (the advertisers). In order to enrich our strategic space even further, and allow advertisers to have multi-dimensional strategic spaces as well, we present also a randomized mechanism termed ``Threshold by Partition Mechanism'' (\TPM). {\TPM} applies to our general setting, \ie, we no longer assume that any capacity is known to the mechanism, and it works as follows: divide the set of mediators uniformly at random into to two sets ($M_1$ and $M_2$) and divide the set of advertiser uniformly at random, as well, into two sets ($A_1$ and $A_2$). Then use the optimal trade for $M_2$ and $A_2$ to produce threshold cost and threshold value that allow BB pricing as well as the needed reduction in trade for $M_1$ and $A_1$. Analogously, use the optimal trade for $M_1$ and $A_1$ to produce threshold cost and threshold value that allow BB pricing as well as the needed reduction in trade for $M_2$ and $A_2$.% The order consideration of the mediators and advertisers in the matching of the reduced sets is done following a random order and independently of the reports received by the mechanism. 

The above description of {\TPM} is not complete since the use of threshold cost and value from one pair $(M_i, A_i)$ to reduce the trade in the other pair might create an unbalanced reduction. To overcome this issue we create two random low priority sets: one of advertisers and the other of mediators. Then, whenever the reduction in trade is unbalanced, we remove additional low priority mediators or advertisers in order to restore balance (which can be done with high probability). The following theorem shows that the above mechanism is IC, IR, BB and provides a non-trivial approximation for the optimal gain from trade. The parameter $\alpha$ is an upper bound, known to the mechanism, on the ratio between the maximum capacity of any player (mediator or advertiser) and the size of the optimal trade.\footnote{The parameters $\gamma$ and $\alpha$ both bound the maximum capacity of the players. Moreover, they are formally related by the formula $\alpha = \gamma/\tau$. We chose to formulate Theorems~\ref{thm:PRM} and~\ref{thm:TPM} in terms of the parameter that the mechanism corresponding to each theorem assumes access to.}

\begin{theorem} \label{thm:TPM}
{\TPM} is BB, IR, IC and $(1 - 28\sqrt[3]{\alpha} - 20e^{-2/\sqrt[3]{\alpha}})$-competitive.
\end{theorem}

We note that {\TPM} is universally truthful, \ie , its IC property holds for every given choice of the random coins of the mechanism. 

One drawback of our results is that the competitive ratios guaranteed by Theorems~\ref{thm:PRM} and~\ref{thm:TPM} are non-trivial only when no single advertiser or mediator has a large market power and the mechanism has access to a good bound on the maximum market power of any player. From a practical point of view we believe these assumption are both plausible. The number of agents using any given real life ad-exchange is usually very large, and the mechanism can use the large quantity of historical data available to it to estimate the bound it needs. From a more theoretical point of view, the impossibility result of~\cite{MS83} shows that no non-trivial competitive ratio can be achieved when one advertiser and one mediator control all the trade. This suggests, although we are unable to prove it formally, that the competitive ratio \emph{must} deteriorate as a single advertiser or mediator gains more and more market power (\ie, when $\gamma$ and $\alpha$ increase).

\subsection{Related Work}
\label{subsec:related-work}
%In this paper we design exchange mechanisms with mediators who act as intermediaries between a set of the players and the mechanism, and motivate this setting with a predicted future evolution of the online advertising market. Thus,
From a motivational point of view the most closely related literature to our work consists of works that involve mediators and online advertising markets, such as \cite{AMT09,FMMP10,SGP14}. These works differ from ours in two crucial points. First, despite being motivated by the online display ads network exchange, the models studied by these works are actually auctions (\ie, one-sided mechanisms). Thus, they need not deal with the challenges and impossibility integrated by the double-sided structure of our market and the requirement to keep it from running into a deficit. Second, our focus is maximization of the gain from trade,
%requiring dominant strategies by the players (users) when their associated mediator is truthful, 
unlike the above works which focus on revenue maximization. 

Another related work involving both markets and mediators studies the phenomenon of markets in which individual buyers and sellers trade through intermediaries, who determine prices via strategic considerations \cite{BEKT07}. %This work models this phenomenon using a game in which buyers, sellers and traders engage in trade on a graph representing the access each buyer and seller has to the traders. % \cite{BEKT07}'s buyers have known capacity of 1.
An essential difference between the model of \cite{BEKT07} and our model is that \cite{BEKT07} does not assume private values for the players, and therefore, the impossibility of \cite{MS83} does not apply in its model.% Thus, \cite{BEKT07} can simultaneously achieve optimal gain from trade, IC, IR and BB and do not face our challenges or attempt to solve them. 

%In general, our work relates to the study of mediators. Similarly to \cite{BFT16}, our setting allows multiple mediators, each having his own captive audience, which must play the game through the mediator. This is quite different from the typical work on mediators in which a mediator serves as an arbitration device: the players are not a captive audience, and each of them can decide to participate in the game directly or work through the mediator \cite{AMT09,MT09}. We refer the reader to~\cite{BFT16} for more detailed discussion of related work relevant to mediators of the type seen in our model.
%Similarly to \cite{BFT16}, in our setting there are multiple mediators, each having his own captive audience, which must play the game through the mediator. Moreover, the mediators are strategic players who are interested in maximizing their utility.  As our mediators are not revenue maximizing but rather utility maximizing, one can view them as collusion sets. In this context, a relevant work is \cite{LST02} which studies game theoretic aspects of bidding clubs in which ``collusion devices'' (cartels) are strategically created in a first price auction. Unlike \cite{LST02}'s model, in our setting the partition of users to mediators is pre-determined and our focus is on mechanism design given that fact.

Last but not least is the literature on trade reduction and multi-sided markets. Deterministic mechanisms using trade reduction as a mean to achieve IC, IR and BB were described for various settings~\cite{M92,CM06,BW03,GGP07,RCJZ01,BNP09}. Moreover, for a variant of the setting of~\cite{M92,BNP09}, \cite{HHA16} obtained a randomized mechanism achieving IC, IR and strong budget balance (\ie, it is BB and leaves no surplus for the market maker). %Two examples for such settings are: a double-sided auction (with or without transaction costs for trading pairs) where multiple strategic sellers offer identical goods to multiple strategic buyers, and each buyer (seller) is interested in (selling) a single good (see~\cite{M92}, who was the first to introduce a trade reduction procedure); and a supply chain with suppliers having unique manufacture technologies, such that each supplier effectively has a monopoly on her goods (see, \eg,~\cite{CM06,BW03}). Additional examples for such settings can be found in~\cite{GGP07,RCJZ01}. %and a distributed market, multiple markets connected by edges that allow the exchange of excess supply or demand. Each edge may optionally have a cost such as a usage fee or transport cost \cite{RCJZ01}. The fifth such mechanism is: A combinatorial market with buyers, each desiring a single heterogeneous bundle (e.g., \cite{GGP07}).
The mutual grounds of all these settings is that all players participating in the trade have a single dimension strategic space. This idea was captured by \cite{GGP07} which provided a single trade reduction procedure applicable to all the above settings. In addition, \cite{GGP07} also defined a class of problems that can be solved by its suggested trade reduction procedure. Essentially this classification is based on partitioning the players participating in the trading set into equivalence classes.

As pointed out in the previous subsection, both our presented mechanisms extend significantly on the existing trade reduction literature. More specifically, even when all advertisers have known equal capacities (while mediators can still have a variable number of users), fitting our model into the classification of \cite{GGP07} still requires each mediator to have his own equivalence class (because a mediator with many users can always replace a mediator with a few users within a trading set, but the reverse is often not true).
% in the trading set as all advertisers still remain in the trading set. However the relations between the mediators are not symmetric as the mediator with few users can not necessarily replace the mediator with many users in his trading set. It follows that in our general setting each mediator and each advertiser need to be in his and her own equivalent class and therefore 
It follows that \cite{GGP07}'s trade reduction procedure might remove all the trade, and thus, achieves only a trivial gain from trade approximation.% More formally, \cite{GGP07}'s procedure is $1-\frac{n q \gamma}{\tau}$-competitive, where $n$ is the number of equivalent classes and $q$ is the maximum frequency in which an equivalent class representatives appears in a trading set. In our setting $n = |M| + 1$ and $q \geq 1$.      

%% file: SecModel.tex
\section{Notation and Basic Observations}

We begin this section with a more formal presentation of our model. Our model consists of a set $P$ of users, a set $M$ of mediators and a set $A$ of advertisers. Each user $p \in P$ has a non-negative cost $c(p)$ which she has to be paid if she is assigned to an advertiser. %Users are non-strategic players and each assigned user $p \in P $ is paid her cost $c(p)$.
The users are partitioned among the mediators, and we denote by $P(m) \subseteq P$ the set of users associated with mediator $m \in M$ (\ie, the sets $\{P(m) \mid m \in M\}$ form a disjoint partition of $P$). The utility of a mediator $m \in M$ is the amount he is paid minus the total cost he has to forward to his assigned users; hence, if $x(p) \in \{0, 1\}$ is an indicator for the event that user $p \in P(m)$ is assigned and $t$ is the payment received by $m$, then the utility of $m$ is $t - \sum_{p \in P(m)} x(p) \cdot c(p)$. %\footnote{If the mediator $m$ reports more users than he actually has, then the mechanism might try to assign users of $m$ that do not exist. When this happens the lie of $m$ is revealed, which might affect his future reputation. Thus, we assume that in this case the utility of $m$ is negative.}
Finally, each advertiser $a \in A$ has a positive capacity $u(a)$, and she gains a non-negative value $v(a)$ from every one of the first $u(a)$ users assigned to her; thus, if advertiser $a$ is assigned $n \leq u(a)$ users and has to pay $t$ then her utility is $n \cdot v(a) - t$.

A mechanism for our model accepts reports from the advertisers and mediators, and based on these reports outputs an assignment of users to advertisers. In addition, the mechanism also decides how much to charge (pay) each advertiser (mediator). The objective of the mechanism is to output an assignment of users to advertisers that maximizes the gain from trade.
%A mechanism for this model knows the mediators and the advertisers, but has no knowledge about their parameters or about the users. The objective of the mechanism is to find an assignment of users to advertisers maximizing the \emph{gain from trade}. For that purpose, the mechanism gets reports from the advertisers and mediators. Each advertiser reports her capacity and value. The reports of the mediators are formed in a slightly more involved way. Each user reports her cost to her mediator, and based on these reports each mediator reports the number of his users and their costs to the mechanism. The users, mediators and advertisers are all strategic, and thus, free to send incorrect reports. In other words, an advertiser may report incorrect capacity and value, a user may report an incorrect cost and a mediator may report any number of users and associate with each one of them an arbitrary cost.

%We say that a user is \emph{truthful} if she reports her true cost. Similarly, an advertiser is \emph{truthful} if she reports her true capacity and value. A mediator is \emph{truthful} if he reports his true number of users, and for every user he reports a cost equal to the cost reported by this user. In other words, a truthful mediator should simply forward the information he received from his users.

For ease of the presentation, it is useful to associate a set $B(a)$ of $u(a)$ slots with each advertiser $a \in A$. We then think of the users as assigned to slots instead of directly to advertisers. Formally, let $B$ be the set of all slots (\ie, $B = \bigcup_{a \in A} B(a)$), then an assignment is a set $S \subseteq B \times P$ in which no user or slot appears in more than one ordered pair. We say that an assignment $S$ assigns a user $p$ to slot $b$ if $(p, b) \in S$. Similarly, we say that an assignment $S$ assigns user $p$ to advertiser $a$ if there exists a slot $b \in B(a)$ such that $(p, b) \in S$. It is also useful to define values for the slots. For every slot $b$ of advertiser $a$, we define the value $v(b)$ of $b$ as equal to the value $v(a)$ of $a$. Using this notation, the gain from trade of an assignment $S$ can be stated as
\[
	\GfT(S)
	=
	\sum_{(p, b) \in S} [v(b) - c(p)]
	\enspace.
\]

In addition to the above notation, we would like to define two additional shorthands that we use occasionally. Given a set $A' \subseteq A$ of advertisers, we denote by $B(A') = \bigcup_{a \in A'} B(a)$ the set of slots belonging to advertisers of $A'$. Similarly, given a set $M' \subseteq M$ of mediators, $P(M') = \bigcup_{m \in M'} P(m)$ is the set of users associated with mediators of $M'$.

\subsection{Comparison of Costs and Values} \label{sec:comparison}

The presentation of our mechanisms is simpler when the values of slots and the costs of users are all unique. Clearly, this is extremely unrealistic as all the slots of a given advertiser have the exact same value in our model. Thus, we simulate uniqueness using a tie-breaking rule. The rule we assume works as follows:
\begin{compactitem}
	\item The mechanism chooses an arbitrary order $\sigma$ on the mediators and advertisers. It is important that this order is chosen independently of the reports received by the mechanism. The mechanism then uses this order to break ties when comparing users to slots and when comparing between users (slots) associated with different mediators (advertisers). For example, when comparing the cost of user $p$ with the value of a slot $b$, the mechanism breaks ties in favor of $p$ if and only if the mediator of $p$ appears earlier than the advertiser of $b$ in $\sigma$.
	%\begin{compactitem}
		%\item When comparing the cost of user $p$ with the value of a slot $b$, the mechanism breaks ties in favor of $p$ if and only if the mediator of $p$ appears earlier than the advertiser of $b$ in $\sigma$.
		%\item When comparing the costs of two users $p_1$ and $p_2$ associated with \emph{different} mediators, the mechanism breaks ties in favor of $p_1$ if and only if the mediator of $p_1$ appears earlier than the mediator of $p_2$ in $\sigma$.
		%\item When comparing the values of two slots $b_1$ and $b_2$ associated with \emph{different} advertisers, the mechanism breaks ties in favor of $b_1$ if and only if the advertiser of $b_1$ appears earlier than the advertiser of $b_2$ in $\sigma$.
	%\end{compactitem}
	\item We assume that the report of every mediator induces some order on the set of users of this mediators. The mechanism uses this order to break ties between the costs of users belonging to the same mediator.% More specifically, when comparing two users $p_1$ and $p_2$ of the same mediator, the mechanism breaks the tie according  to the order induced by the report of this mediator.
	\item Finally, since the slots of a given advertiser are all identical and non-strategic (recall that slots were introduced into the model just for the purpose of simplifying the presentation), any method can be used for tie-breaking between the slots of a given advertiser.%The mechanism chooses an arbitrary order for the slots of each advertiser. Since the slots of a given advertiser are all identical, it is not important how this order is selected. When comparing values of two slots belonging to the same advertiser (which always have the same value), the mechanism breaks the tie according to the above order.
\end{compactitem}

In the rest of this paper when costs/values are compared, unless it is explicitly specified that they are compared as numbers, the comparison is assumed to use the above tie breaking rule. Note that this assumption implies that two values (costs) are equal if and only if they belong to the same slot (user). We now \inArXiv{prove}\inConference{give} a useful observation that follows from the way we defined the tie-breaking rule.%\inConference{ (a proof of the observation can be found in Appendix~\ref{app:missing_proofs})}.

\begin{observation} \label{obs:slots_consecutive}
When the slots are ordered in an increasing (or decreasing) value order, all the slots of a single advertiser are always consecutive.
\end{observation}
\proofSlotsConsecutive

%% file: SecCanonical.tex
\subsection{Canonical Assignment}

Given a set $B' \subseteq B$ of slots and a set $P' \subseteq P$ of users, the canonical assignment $S_c(P', B')$ is the assignment constructed by the following process. First, we order the slots of $B'$ in a decreasing value order $b_1, b_2, \dotsc, b_{|B'|}$ and the users of $P'$ in an increasing cost order $p_1, p_2, \dotsc, p_{|P'|}$. Then, for every $1 \leq i \leq \min\{|B'|, |P'|\}$, the canonical assignment $S_c(B', P')$ assigns user $p_i$ to slot $b_i$ if and only if $v(b_i) > c(p_i)$. The canonical assignment is an important tool used frequently by the mechanisms we describe in the next sections. In some places we refer to the user or slot at location $i$ of a canonical solution $S_c(P', B')$. By using this expression we mean user $p_i$ or slot $b_i$, respectively. Additionally, the term $|S_c(P, B)|$ is used very often in our proofs, and thus, it is useful to define the shorthand $\tau = |S_c(P, B)|$.

The following lemma shows that the above definition of $\tau$ is consistent with the use of $\tau$ in Section~\ref{ssc:contribution} as the size of the optimal trade\inConference{ (a proof of the lemma can be found in Appendix~\ref{app:missing_proofs})}.

\begin{lemma} \label{lem:canonical_optimal}
The canonical assignment $S_c(P', B')$ maximizes $\GfT(S_c(P', B'))$ among all the possible assignments of users of $P'$ to slots of $B'$.
\end{lemma}
\inArXiv{\proofCanonicalOptimal}

%% file: SecDeterministic.tex
\section{Deterministic Mechanism} \label{sec:deterministic}

In this section we describe the deterministic mechanism ``Price by Removal Mechanism'' (\PRM) for our model. Recall that {\PRM} assumes public knowledge of the advertisers' capacities. Accordingly, we assume throughout this section that the capacities of the advertisers are common knowledge (or that the advertisers are not strategic about them). We also assume that {\PRM} has access to a value $\gamma \geq 1$ such that:
\[
	u(a) \leq	\gamma \quad \forall\; a \in A
	\qquad \text{and}	\qquad
	|P(m)| \leq	\gamma \quad \forall\; m \in M
	\enspace.
\]
In other words, $\gamma$ is an upper bound on how large can be the capacity of an advertiser or the number of users of a mediator. Informally, $\gamma$ can be understood as a bound on the importance every single advertiser or mediator can have.%\footnote{The definition of $\gamma$ raises a concern regarding the case that a mediator decides to deviate, and reports a number of users larger than $\gamma$. This can be discouraged by making {\PRM} ignore mediators reporting more users than $\gamma$, which guarantees that such deviations are never beneficial. However, for simplicity we assume in this section that mediators are only allowed to report a number of users smaller than or equal to $\gamma$.}

A description of {\PRM} is given as Mechanism~\ref{mch:PRM}. Notice that Mechanism~\ref{mch:PRM} often refers to parameters of the model that are not known to the mechanism (\ie, values of advertiser, the number of users of mediators and the costs of users). Whenever this happens, this should be understood as referring to the reported values of these parameters.

\vspace{2mm}
\noindent \begin{minipage}{\textwidth}
\captionsetup{type=mechanism}
\noindent \rule{\linewidth}{0.8pt}
\vspace{-6mm}\captionof{mechanism}{Price by Removal Mechanism (\PRM)}\label{mch:PRM}
\noindent \rule{\linewidth}{0.8pt}
\vspace{-5.5mm}
\begin{compactenum}[\bfseries 1.]
	\item For every mediator $m \in M$, if the canonical assignment $S_c(P \setminus P(m), B)$ is of size more than $4\gamma$, denote by $p_m$ the user at location $|S_c(P \setminus P(m), B)| - 4\gamma$ of the canonical assignment $S_c(P \setminus P(m), B)$, and let $c_m$ be the cost of $p_m$. Otherwise, set $c_m$ to $-\infty$.
\end{compactenum}
\end{minipage}
{
\hphantom{1}
\setlength{\plitemsep}{3pt}
\begin{compactenum}[\bfseries 1.]
	\setcounter{enumi}{1}
	\item For every mediator $m \in M$, let $\hat{P}(m)$ be the set of users of mediator $m$ whose cost is less than $c_m$. Intuitively, $\hat{P}(m)$ is the set of users of mediator $m$ that the mechanism tries to assign to advertisers.
	
	\item Assign the users of $\bigcup_{m \in M} \hat{P}(m)$ to the advertisers using a VCG auction. More specifically, the users of $\bigcup_{m \in M} \hat{P}(m)$ are the items sold in the auction, and the bidders are the advertisers of $A$ plus a dummy advertiser $a_d$ whose value and capacity are $v(a_d) = \max_{m \in M} c_m$ and $u(a_d) = \sum_{m \in M} |\hat{P}(m)|$, respectively. %It is important that in case of a tie between $v(a_d)$ and the value of a non-dummy advertiser, the VCG auction breaks the tie in favor of the non-dummy advertiser (\ie, the value of the non-dummy advertiser is treated as larger than $v(a_d)$).

	\item Charge every non-dummy advertiser by the same amount she is charged (as a bidder) by the VCG auction.
\end{compactenum}
}
\vspace{2mm}
\noindent\begin{minipage}{\textwidth}
\begin{compactenum}[\bfseries 1.]
	\setcounter{enumi}{4}
	\item For every user $p$ assigned by the VCG auction, pay $c_m$ to the mediator $m$ of $p$.\footnotemark
\end{compactenum}
\vspace{-3mm}\noindent \rule{\linewidth}{0.8pt}
\end{minipage}
\footnotetext{Note that $m$ is BB as he forwards to each of his assigned users her cost---which is less than $c_m$.}
\vspace{2mm}

\noindent \textbf{Remark:} It can be shown that the existence of the dummy advertiser never affects the behavior of Mechanism~\ref{mch:PRM}, and thus, one can safely omit it from the mechanism. Nevertheless, we keep this advertiser in the above description of the mechanism since its existence simplifies our proof that the mechanism is BB.

Let us recall our result regarding {\PRM}\inConference{ (the proof of Theorem~\ref{thm:PRM} is deferred to Appendix~\ref{app:deterministic})}.

\begin{reptheorem}{thm:PRM}
{\PRM} is BB, IR, IC and $\left(1 - \frac{5\gamma}{\tau}\right)$-competitive.
\end{reptheorem}

{\PRM} improves over the trade reduction procedure of~\cite{GGP07} by guarantying a non-trivial competitive ratio even in the presence of players (mediators) with a multi-dimensional strategic space. However, the competitive ratio guaranteed by {\PRM} is slightly worse than the competitive ratio guaranteed by the procedure of~\cite{GGP07} when the players have single dimensional strategic spaces (which is $1 - 2/\tau$). The next paragraph gives an intuitive explanation why it does not seem possible to improve the competitive ratio of {\PRM} to match the competitive ratio of $1 - 2/\tau$ guaranteed by \cite{GGP07}'s procedure.

Observe that {\PRM} needs to maintain IC for both mediators and advertisers. In order to maintain IC for the mediators the price for each mediator has to be computed while all his users are reduced from the trade. This is why the \textbf{real} reduction in trade might be larger by up to $\gamma$ compared to the reduction $4\gamma$ explicitly stated by {\PRM}, which explains the gap between the competitive ratio of $1 - 5\gamma/\tau$ and the explicit reduction of $4\gamma$ in the mechanism. It remains to understand why the explicit reduction is set to $4\gamma$. The most significant difficulty in guaranteeing IC for the advertisers is that an advertiser might change her report in order to affect the costs $\{c_m\}_{m \in M}$ of the mediators, and through them, manipulate the items offered in the VCG auction. {\PRM} tackles this difficulty by guaranteeing that advertisers performing such manipulations are assigned no users. Since the users of mediator $m$ are removed when $c_m$ is calculated, securing this guarantee requires that advertisers having a slot in one of the last $\gamma$ locations of the optimal trade are assigned no users. As an advertiser might have up to $\gamma$ slots, this translates into a requirement that an advertiser whose earliest slot is in one of the last $2\gamma$ locations of the optimal trade is assigned no users. Moreover, to simplify the proof our analysis in fact requires that an advertiser whose earliest slot is in one of the last $3\gamma$ locations of the optimal trade is assigned no users. Taking into account, again, the fact that the users of mediator $m$ are removed when $c_m$ is calculated, guaranteeing the last property requires a trade reduction of $4\gamma$. To summarize, maintaining the advertisers' IC imposes an explicit trade reduction of $3\gamma$ that we present as $4\gamma$ for the sake of simplicity, and maintaining the mediators' IC implies that the real trade reduction is larger by up to $\gamma$ compared to the explicit one. Thus, a competitive ratio of $1-\frac{4\gamma}{\tau}$ seems to be inevitable in order to allow multi-dimensional strategic spaces.

\inArXiv{\input{AppDeterministic}}

%% file: AppDeterministic.tex
\inConference{\section{Proof of Theorem~\ref{thm:PRM}} \label{app:deterministic}}

\inArXiv{In the rest of}\inConference{In} this section we prove Theorem~\ref{thm:PRM}. We begin with some basic properties of {\PRM}. It is important to remember while reading the next proofs that the constructions of $S_c(P, B)$ and $S_c(P \setminus P(m), B)$ order the slots of $B$ in the same order, and that the order of the users of $P \setminus P(m)$ in the construction of $S_c(P \setminus P(m), B)$ is obtained by removing the users of $P(m)$ from the order used by the construction of $S_c(P, B)$.

\begin{observation} \label{obs:med_removal_shortens}
For every mediator $m \in M$, $|S_c(P \setminus P(m), B)| \leq \tau$.
\end{observation}
\begin{proof}
Assume towards a contradiction that $|S_c(P \setminus P(m), B)| > \tau$, and let $b'_{\tau + 1}$ and $p'_{\tau + 1}$ be the slot and user at location $\tau + 1$ of $S_c(P \setminus P(m), B)$, respectively. Since $p'_{\tau + 1}$ is assigned by $S_c(P \setminus P(m), B)$ to $b'_{\tau + 1}$, it must be that $v(b'_{\tau +1}) > c(p'_{\tau + 1})$.

On the other hand, let $b_{\tau + 1}$ and $p_{\tau + 1}$ be the slot and user at location $\tau + 1$ of $S_c(P, B)$, respectively. Clearly $b_{\tau + 1}$ and $b'_{\tau + 1}$ are exactly the same slot. Moreover, $p'_{\tau + 1}$ is either equal to $p_{\tau + 1}$ or appears in $S_c(P, B)$ in a larger location than $p_{\tau + 1}$. Hence,
\[
	v(b_{\tau + 1})
	=
	v(b'_{\tau + 1})
	>
	c(p'_{\tau + 1})
	\geq
	c(p_{\tau + 1})
	\enspace,
\]
which is a contradiction since $p_{\tau + 1}$ is not assigned to $b_{\tau + 1}$ by $S_c(P, B)$.
\end{proof}

Using the last observation we can now prove the following lemma.

\begin{lemma} \label{lem:subset_users}
For every mediator $m \in M$ the following is true:
\begin{compactitem}
	\item When $c_m \neq -\infty$, $p_m$ is assigned by the canonical assignment $S_c(P, B)$.% Moreover, in this case $p_m$ is the user at location $i$ of $S_c(P, B)$ for some $1 \leq i \leq \tau - 3\gamma$.
	\item All the users of $\hat{P}_m$ are assigned by $S_c(P, B)$.
\end{compactitem}
\end{lemma}
\begin{proof}
First, let us explain why the second claim we need to prove follows from the first one. If $c_m = -\infty$, then $\hat{P}_m = \varnothing$ and the second claim is trivial. Otherwise, the first claim states that $p_m$ is assigned by the canonical assignment $S_c(P, B)$. By the definition of a canonical assignment, all the users having a lower cost than $c_m = c(p_m)$ are assigned by $S_c(P, B)$ as well. The second claim now follows since $\hat{P}_m$ contains only users having a lower cost than $c_m$.

It remains to prove the first claim. The claim is trivial when $c_m = -\infty$, thus, we assume $c_m \neq -\infty$. %, which implies, by the definition of $c_m$ and Observation~\ref{obs:med_removal_shortens}, $\tau \geq |S_c(P \setminus P(m), B)| > 4\gamma$.
Let $F$ be the set of users that are assigned by $S_c(P, B)$ but do not belong to $P(m)$. Clearly the users of $F$ take locations $1$ to $|F|$ in $S_c(P \setminus P(m), B)$. On the other hand, we can lower bound the size of $F$ by $\tau - |P(m)| \geq \tau - \gamma$. Thus, locations $1$ to $\tau - \gamma$ of $S_c(P \setminus P(m), B)$ are all occupied by users assigned by $S_c(P, B)$.

By Observation~\ref{obs:med_removal_shortens} and the fact that $c_m \neq -\infty$, $p_m$ is chosen by {\PRM} as the user located at location $|S_c(P \setminus P(m), B)| - 4\gamma \leq \tau - 4\gamma$ of $S_c(P \setminus P(m), B)$. This implies that $p_m$ is assigned by $S_c(P, B)$ since $\tau - 4\gamma \leq \tau - \gamma$ (and all the users at locations $1$ to $\tau - \gamma$ of $S_c(P \setminus P(m), B)$ are assigned by $S_c(P, B)$). %This implies two things: first $p_m$ is assigned by $S_c(P, B)$ since $\tau - 4\gamma \leq \tau - \gamma$ (and all the users at locations $1$ to $\tau - \gamma$ of $S_c(P \setminus P(m), B)$ are assigned by $S_c(P, B)$). Second, the $3\gamma$ users located at locations $\tau - 4\gamma + 1$ to $\tau - \gamma$ of $S_c(P \setminus P(m), B)$ are assigned by $S_c(P, B)$ and have a higher cost than $p_m$, and thus, $p_m$ cannot appear in one of the locations $\tau - 3\gamma + 1$ to $\tau$ of $S_c(P, B)$.
\end{proof}

\begin{corollary} \label{cor:all_assigned}
{\PRM} is BB and assigns all the users of $\bigcup_{m \in M} \hat{P}_m$.
\end{corollary}
\begin{proof}
Let $M'$ be the subset of mediators having a finite $c_m$. Formally, $M' = \{m \in M \mid c_m \neq -\infty\}$. If $M' = \varnothing$, then no users are assigned by {\PRM} and $\hat{P}_m = \varnothing$ for every mediator $m \in M$, which makes the corollary trivial. Hence, we assume from now on $M' \neq \varnothing$. We define now $p_\tau$ as the user at location $\tau$ in $S_c(P, B)$. Clearly, $c(p_\tau)$ upper bounds the cost of any user assigned by $S_c(P, M)$. On the other hand, by Lemma~\ref{lem:subset_users}, for every $m \in M'$ the user $p_m$ is assigned by $S_c(P, B)$. Hence, $c(p_\tau) \geq \max_{m \in M'} c(p_m) = \max_{m \in M'} c_m$.

Let $B'$ be the set of slots assigned users by $S_c(P, B)$, and let $b_\tau$ be the slot at location $\tau$ in $S_c(P, B)$. Since $p_\tau$ is assigned to $b_\tau$ by $S_c(P, B)$, we get: $v(b_\tau) > c(p_\tau) \geq \max_{m \in M'} c_m$. Moreover, $v(b_\tau)$ lower bounds the value of any slot in $B'$, and thus, all the slots of $B'$ have values larger than $\max_{m \in M'} c_m$. Combining this with the observation that $|B'| = \tau$, we get that there exists a set $A' \subseteq A$ of advertisers with the following properties:
\begin{compactitem}
	\item The advertisers of $A'$ all have values larger than $\max_{m \in M'} c_m = \max_{m \in M} c_m$, and thus, larger than the value of the dummy advertiser $a_d$, and larger than the cost of any user in $\bigcup_{m \in M'} \hat{P}_m = \bigcup_{m \in M} \hat{P}_m$.
	\item The advertisers of $A'$ have a total capacity at least $\tau \geq |\bigcup_{m \in M} \hat{P}_m|$, where the inequality holds since Lemma~\ref{lem:subset_users} guarantees that the users of $\bigcup_{m \in M} \hat{P}_m$ are all assigned by $S_c(P, M)$.
\end{compactitem}

The two above properties imply together that the VCG auction (and thus, also {\PRM}) assigns all the users of $\bigcup_{m \in M} \hat{P}_m$ to real advertisers. As no user is assigned to the dummy advertiser, and the dummy advertiser has a value of $\max_{m \in M} c_m$, the VCG auction charges the real advertisers at least $\max_{m \in M} c_m$ for every user assigned to them. On the other hand, $\max_{m \in M} c_m$ upper bounds the payment a mediator receives from {\PRM} for a single assigned user, and thus, {\PRM} is budget balanced.
\end{proof}

\subsection{The Competitive Ratio of {\PRM}}

In this section we analyze the competitive ratio of {\PRM}, and show that it is at least $1 - 5\gamma/\tau$. If $5\gamma \geq \tau$, then this is trivial. Thus, we assume throughout this section $5\gamma < \tau$.

\begin{observation} \label{obs:stay_assigned}
Given a mediator $m \in M$, every user $p \not \in P(m)$ which is assigned by the canonical assignment $S_c(P, B)$ is also assigned by the canonical assignment $S_c(P \setminus P(M), B)$.
\end{observation}
\begin{proof}
Let $j'$ and $j$ be the locations of $p$ in $S_c(P \setminus P(m), B)$ and $S_c(P, B)$, respectively. Clearly $j' \leq j$, and thus, the slot at location $j'$ in $S_c(P \setminus P(m), B)$ has a value larger or equal to the value of the slot at location $j$ in $S_c(P, B)$. This means that $p$ must be assigned by $S_c(P \setminus P(m), B)$ since it is assigned by $S_c(P, B)$.
\end{proof}

Let us denote by $c_e$ the cost of the user at location $\tau - 5\gamma$ of the canonical assignment $S_c(P, B)$. Using the last observation we get:
\begin{lemma} \label{lem:c_e_bound}
For every mediator $m \in M$, $c_m \geq c_e$.
\end{lemma}
\begin{proof}
Let $L$ be the set of the users at locations $\tau - 5\gamma$ to $\tau$ of $S_c(P, B)$. Clearly all the users of $L$ have a cost equal or larger than $c_e$. On the other hand, observe that by Observation~\ref{obs:stay_assigned} all the users of $L \setminus P(m)$ are assigned by $S_c(P \setminus P(m), B)$. Hence, the number of users of $L$ that are assigned by $S_c(P \setminus P(m), B)$ is at least:
\[
	|L \setminus P(m)|
	\geq
	|L| - |P(m)|
	\geq
	(5\gamma + 1) - \gamma
	=
	4\gamma + 1
	\enspace.
\]

Recall that $p_m$ is selected as the user at location $|S_c(P \setminus P(m), B)| - 4\gamma$ in $S_c(P \setminus P(m), B)$. This, together with the observation that $S_c(P, \setminus P(m), B)$ assigns at least $4\gamma + 1$ users of $L$, implies that one of the following must be true: either $p_m$ is a user of $L$, or there exists a user of $L$ that has a smaller location in $S_c(P, \setminus P(m), B)$ than $p_m$, and thus, has a lower cost than $p_m$. In both cases the lemma follows.\footnote{In fact, it can be shown that the second case never happens.}
\end{proof}

\begin{corollary}
{\PRM} is $\left(1 - \frac{5\gamma}{\tau}\right)$-competitive.
\end{corollary}
\begin{proof}
Let $\bar{P}_e$ be the set of users at locations $1$ to $\tau - 5\gamma$ of the canonical assignment $S_c(P, B)$. Consider an arbitrary user $p \in \bar{P}_e$, and let $m$ be her mediator. Lemma~\ref{lem:c_e_bound} and the definition of $\bar{P}_e$ imply together $c(p) \leq c_e \leq c_m$. On the other hand, $c_m = c(p_m)$ is the cost of a user $p_m \not \in P(m)$, and thus, we cannot have $c(p) = c_m$ since the tie-breaking rule defined in Section~\ref{sec:comparison} guarantees that no two users have the same cost. Hence, $c(p) < c_m$.

Recall that $\hat{P}_m$ is defined as the set of all users of $P(m)$ whose cost is less than $c_m$. Together with the previous observation, this gives $\bar{P}_e \cap P(m) \subseteq \hat{P}_m$ for every mediator $m \in M$, and therefore, $\bar{P}_e \subseteq \bigcup_{m \in M} \hat{P}_m$. Notice now that Lemma~\ref{lem:subset_users} states that all the users of $\bigcup_{m \in M} \hat{P}_m$ are assigned by {\PRM}. Hence, in particular, the users of the subset $\bar{P}_e$ are all assigned by {\PRM}.

Denote by $S$ the assignment produced by {\PRM}. Then:
\begin{align} \label{eq:gain_from_trade}
	\GfT(S)
	={} &
	\sum_{(p, b) \in S} [v(b) - c(p)]
	=
	\sum_{(p, b) \in S} [v(b) - \max\nolimits_{m \in M} c_m] + \sum_{(p, b) \in S} [\max\nolimits_{m \in M} c_m - c(p)]\\
	\geq{} &
	\sum_{(p, b) \in S} [v(b) - \max\nolimits_{m \in M} c_m] + \sum_{p \in \bar{P}_e} [\max\nolimits_{m \in M} c_m - c(p)] \nonumber
	\enspace,
\end{align}
where the inequality holds by the above discussion and the observation that every user assigned by {\PRM} belongs to $\cup_{m \in M} \hat{P}_m$, and thus, has a cost of at most $\max_{m \in M} c_m$. 

The VCG auction assigns the users in a way that maximizes the total value of the advertisers from the assigned users. This means that the $|S|$ users assigned by $S$ are assigned to the slots at locations $1$ to $|S|$ of $S_c(P, B)$. Let us denote this set of slots by $B_S$. Additionally, let us denote by $\bar{B}_e$ the set of slots at locations $1$ to $|\bar{P}_e|$ in $S_c(P, B)$. Then,
\[
	\sum_{(p, b) \in S} [v(b) - \max\nolimits_{m \in M} c_m]
	=
	\sum_{b \in B_S} [v(b) - \max\nolimits_{m \in M} c_m]
	\geq
	\sum_{b \in \bar{B}_e} [v(b) - \max\nolimits_{m \in M} c_m]
	\enspace,
\]
where the inequality holds due to two observations: first $\bar{B}_e$ is a subset of $B_S$ since at least $|\bar{B}_e| = |\bar{P}_e|$ users are assigned by $S$. Second, the dummy advertiser $a_d$ has a value of $\max_{m \in M} c_m$, and thus, every slot of $B_S$ has a value of at least $\max_{m \in M} c_m$ since the VCG auction preferred assigning a user to this slot over assigning it to $a_d$.

Plugging the last inequality into Inequality~\eqref{eq:gain_from_trade} gives:
\begin{align*}
	\GfT(S)
	\geq{} &
	\sum_{b \in \bar{B}_e} [v(b) - \max\nolimits_{m \in M} c_m] + \sum_{p \in \bar{P}_e} [\max\nolimits_{m \in M} c_m - c(p)]
	=
	\sum_{b \in \bar{B}_e} v(b) - \sum_{p \in \bar{P}_e} c(p)\\
	\geq{} &
	\frac{\tau - 5\gamma}{\tau} \cdot \sum_{(p, b) \in S_c(P, B)} \mspace{-18mu} v(b) - \frac{\tau - 5\gamma}{\tau} \cdot \sum_{(p, b) \in S_c(P, B)} \mspace{-18mu} c(p)
	%=
	%\frac{\tau - 5\gamma}{\tau} \cdot \sum_{(p, b) \in S_c(P, B)} [v(b) - c(p)]
	=
	\frac{\tau - 5\gamma}{\tau} \cdot \GfT(S_c(P, B))
	\enspace,
\end{align*}
where the second inequality holds since $\bar{P}_e$ ($\bar{B}_e$) contains, by definition, the $\tau - 5\gamma$ users (slots) with the lowest costs (highest values) among the $\tau$ users (slots) assigned by the canonical assignment $S_c(P, B)$. The corollary now follows from the last inequality since $S_c(P, B)$ is an assignment of users from $P$ to slots of $B$ maximizing the gain from trade.
\end{proof}

\subsection{The Incentive Properties of {\PRM}}

In this section we prove the incentive parts of Theorem~\ref{thm:PRM}. We begin with a lemma analyzing the incentive properties of {\PRM} for mediators.

%\begin{lemma}
%For every user $p$, assuming the mediator $m$ of $p$ is truthful, {\PRM} is IR for $p$, and truthfulness is a dominant strategy for her.
%\end{lemma}
%\begin{proof}
%{\PRM} calculates a threshold $c_m$ based on the reports of advertisers and mediators other than $m$, and thus, $c_m$ is independent of the report of $p$. If $p$ reports a cost larger than $c_m$, then $p$ is not added to $\hat{P}_m$, and thus, {\PRM} does not assign $p$. On the other hand, if $p$ reports a cost smaller than $c_m$, then $p$ is added to $\hat{P}_m$, and by Lemma~\ref{lem:subset_users} she is assigned by {\PRM} to some slot and is paid $c_m$ (more specifically, {\PRM} recommends for $m$ to forward an amount of $c_m$ to $p$, and $m$ does that since he is truthful by our assumption).

%From the above description it is clear that $p$ is always paid, when assigned, at least the cost she reports, and thus, {\PRM} is IR for $p$. Moreover, she is assigned and paid $c_m$ exactly when this is beneficial for her according to her report, and thus, truthfulness is a dominant strategy for $p$.
%\end{proof}

\begin{lemma}
For every mediator $m$, {\PRM} is IR for $m$, and truthfulness is a dominant strategy for him.
\end{lemma}
\begin{proof}
{\PRM} calculates a threshold $c_m$ based on the reports of advertisers and mediators other than $m$, and thus, $c_m$ is independent of the report of $m$. For every user $p \in P(m)$, if $p$'s cost is larger than $c_m$, then $p$ is not added to $\hat{P}_m$, and thus, {\PRM} does not assign $p$. On the other hand, if $p$'s cost is smaller than $c_m$, then $p$ is added to $\hat{P}_m$, and by Lemma~\ref{lem:subset_users} she is assigned by {\PRM} to some slot and $m$ is paid $c_m$.

From the above description it is clear that when $m$ is truthful, he is paid for each assigned user $p \in P(m)$ at least the cost of the user. Thus, {\PRM} is IR for $m$. To see why truthfulness is a dominant strategy for $m$, notice that the utility function of $m$ is maximized when all the users of $m$ whose cost is less than $c_m$ are assigned, and no other user of $m$ is assigned; which is exactly what {\PRM} does when $m$ is truthful.
\end{proof}

Proving the incentive properties of {\PRM} for the advertisers is more involved. We start by proving that {\PRM} is IR for advertisers.

\begin{observation} \label{obs:advertisers_IR}
For every advertiser $a$, {\PRM} is IR for $a$.
\end{observation}
\begin{proof}
{\PRM} uses a VCG auction to assign users to advertisers, and to determine the cost each advertiser has to pay. Since VCG auctions are IR, the payment charged for every advertiser is upper bounded by the total value the advertiser gets from the users assigned to her.
\end{proof}

In the rest of the section we concentrate proving that truthfulness is a dominant strategy for the advertisers under {\PRM} . Let us fix an advertiser $a \in A$. For technical reasons it is convenient to somewhat enhance the strategy space of $a$. Clearly, if we can show that truthfulness is a dominant strategy for $a$ given the enhanced strategy space, then this is also true given the original strategy space.

Section~\ref{sec:comparison} explains how ties between valuations of advertisers are broken using an order $\sigma$. In particular, when comparing the values of slots of two different advertisers with identical values, the values of the slots of the advertiser that appears earlier in $\sigma$ are considered larger. We enhance the strategy space of $a$ by allowing her to report both her value and her location in $\sigma$ compared to the other advertisers and mediators. In other words, the relative order of the other advertisers and mediators in the order $\sigma$ defining the tie-breaking rule is fixed, and $a$ can choose to insert herself at any location within this order.

Given a possible strategy $s$ for $a$, we denote by $B_s$ the set of slots given that $a$ uses this strategy, and by $\ell(s)$ the lowest location in $S_c(P, B_s)$ of a slot of $a$. Note that, since the slots of $a$ appear sequentially in $S_c(P, B_s)$, they must appear at locations $\ell(s)$ to $\ell(s) + u(a) - 1$ in $S_c(P, B_s)$. 

\begin{observation} \label{obs:density}
For every two indices $1 \leq i \leq j \leq |B| - (\gamma - 1)$, if $i + (\gamma - 1) \leq j$, then there exists a strategy $s$ for advertiser $a$ such that $i \leq \ell(s) \leq j$.
\end{observation}
\begin{proof}
$\ell(s) - 1$ is equal to the total capacity of the advertisers appearing before the location set by the strategy $s$ for $a$ in the order $\sigma$. Since the capacity of every advertiser is at most $\gamma$, the values $\ell(s)$ corresponding to two adjacent possible locations for $a$ in $\sigma$ are different by at most $\gamma$. Moreover, if $a$ is inserted by $s$ as the first advertiser in $\sigma$ than $\ell(s) = 1$, and if $a$ is inserted by $s$ as the last advertiser in $\sigma$ than $\ell(s) = |B| - (u(a) - 1) \geq |B| - (\gamma - 1)$. Hence, for every list of consecutive indices of length at least $\gamma$ such that all the indices are between $1$ and $|B| - (\gamma - 1)$, there must be a strategy $s$ for $a$ such that $\ell(s)$ is within this list.
\end{proof}

The next claims analyze some possible strategies for $a$. Let $\bar{s}$ be the strategy of $a$ in which $a$ reports a value of $\infty$, and let $\bar{\tau}$ be the size of the canonical assignment $S_c(P, B_{\bar{s}})$.

\begin{observation} \label{obs:upper_bound_double_change}
For every mediator $m \in M$ and strategy $s$ for $a$, $|S_c(P \setminus P(m), B_s)| \leq \bar{\tau}$.
\end{observation}
\begin{proof}
Assume towards a contradiction that $|S_c(P \setminus P(m), B_s)| > \bar{\tau}$, and let $p_{\bar{\tau} + 1}$ and $b_{\bar{\tau} + 1}$ be the user and slot, respectively, at location $\bar{\tau} + 1$ in $S_c(P \setminus P(m), B_s)$. By our assumption, $p_{\bar{\tau} + 1}$ is assigned to $b_{\bar{\tau} + 1}$, and therefore, $v(b_{\bar{\tau} + 1}) > c(p_{\bar{\tau} + 1})$.

Let us now define $p'_{\bar{\tau} + 1}$ and $b'_{\bar{\tau} + 1}$ to be the user and slot, respectively, at location $\bar{\tau} + 1$ in $S_c(P, B_{\bar{s}})$. Clearly, $p_{\bar{\tau} + 1}$ appears at a location greater or equal to $\bar{\tau} + 1$ in $S_c(P, B_{\bar{s}})$, hence, $c(p_{\bar{\tau} + 1}) \geq c(p'_{\bar{\tau} + 1})$. If $b'_{\bar{\tau} + 1}$ is a slot of $a$, then by definition $v(b'_{\bar{\tau} + 1}) = \infty > v(b_{\bar{\tau} + 1})$. Otherwise, the location of $b'_{\bar{\tau} + 1}$ in $S_c(P \setminus P(m), B_s)$ must be smaller or equal to $\bar{\tau} + 1$ (its location in $S_c(P, B_{\bar{s}})$), and thus, $v(b'_{\bar{\tau} + 1}) \geq v(b_{\bar{\tau} + 1})$.

Combining all the inequalities we proved gives: $v(b'_{\bar{\tau} + 1}) \geq v(b_{\bar{\tau} + 1}) > c(p_{\bar{\tau} + 1}) \geq c(p'_{\bar{\tau} + 1})$, which is a contradiction since $b'_{\bar{\tau} + 1}$ and $p'_{\bar{\tau} + 1}$ have the same location in $S_c(P, B_{\bar{s}})$, but $p'_{\bar{\tau} + 1}$ is not assigned to $b'_{\bar{\tau} + 1}$ by $S_c(P, B_{\bar{s}})$.
\end{proof}

%\begin{proof}
%Assume towards a contradiction that $|S_c(P \setminus P(m), B_s)| > \bar{\tau} + \gamma$. Let $b_{\bar{\tau} + 1}$ be the slot at location $\bar{\tau} + 1$ in $S_c(P, B \setminus B(a))$. Since $a$ has only $u(a) \leq \gamma$ slots in $B_s$, $b_{\bar{\tau} + 1}$ appears at some location $\bar{\tau} < i \leq \bar{\tau} + \gamma + 1$ in $S_c(P \setminus P(m), B_s)$. Hence, by our assumption, $b_{\bar{\tau} + 1}$ is assigned the user $p'$ that appears at location $i$ in $S_c(P \setminus P(m), B_s)$, and therefore, $v(b_{\bar{\tau} + 1}) > c(p')$.
%
%Next, let $p_{\bar{\tau} + 1}$ be the user at location $\bar{\tau} + 1$ in $S_c(P, B \setminus B(a))$. As $p'$ appears at location $i \geq \bar{\tau} + 1$ in $S_c(P \setminus P(m), B_s)$, the location of $p'$ in $S_c(P, B \setminus B(a))$ must be also at least $\bar{\tau} + 1$, and thus, the cost of $p'$ is lower bounded by the cost of $p_{\bar{\tau} + 1}$. More formally, we get: $c(p_{\bar{\tau} + 1}) \leq c(p') < v(b_{\bar{\tau} + 1})$, which is a contradiction since $b_{\bar{\tau} + 1}$ and $p_{\bar{\tau} + 1}$ have the same location in $S_c(P, B \setminus B(a))$, but $p_{\bar{\tau} + 1}$ is not assigned to $b_{\bar{\tau} + 1}$ by $S_c(P, B \setminus B(a))$.
%\end{proof}

Using the last observation we can now prove the following lemma, which considers strategies of $a$ with a relatively large $\ell(s)$.

\begin{lemma} \label{lem:no_users_case}
Given that advertiser $a$ uses a strategy $s$ such that $\ell(s) \geq \bar{\tau} - 3\gamma$, {\PRM} assigns no users to $a$.
\end{lemma}
\begin{proof}
Let $c_{\bar{\tau} - 3\gamma}$ be the cost of the user at location $\bar{\tau} - 3\gamma$ in $S_c(P, B_{\bar{s}})$ if $\bar{\tau} > 3\gamma$, and $-\infty$ otherwise. Our first objective is to show that $c_{\bar{\tau} - 3\gamma} \geq c_m$ for every mediator $m \in M$. Consider an arbitrary mediator $m \in M$. By Observation~\ref{obs:upper_bound_double_change}, $|S_c(P \setminus P(m), B_s)| \leq \bar{\tau}$. There are now two cases to consider:
\begin{compactitem}
	\item If $|S_c(P \setminus P(m), B_s)| \leq 4\gamma$, then $c_m = -\infty \leq c_{\bar{\tau} - 3\gamma}$.
	\item If $|S_c(P \setminus P(m), B_s)| > 4\gamma$, then $c_m$ is the cost of the user at location $|S_c(P \setminus P(m), B_s)| - 4\gamma$ of $S_c(P \setminus P(m), B_s)$. This user must appear in $S_c(P, B_{\bar{s}})$ either at location $(|S_c(P \setminus P(m), B_s)| - 4\gamma) + \gamma \leq \bar{\tau} - 3\gamma$ or at a smaller location. On the other hand, since $\bar{\tau} \geq |S_c(P \setminus P(m), B_s)| > 4\gamma$, $c_{\bar{\tau} - 3\gamma}$ is the cost of the user at location $\bar{\tau} - 3\gamma$ of $S_c(P, B_{\bar{s}})$, and thus, $c_{\bar{\tau} - 3\gamma} \geq c_m$.
\end{compactitem}

The fact that $c_{\bar{\tau} - 3\gamma} \geq \max_{m \in M} c_m$ implies that $\bigcup_{m \in M} \hat{P}_m$ contains only users whose cost is smaller than $c_{\bar{\tau} - 3\gamma}$, and there are only $\bar{\tau} - 3\gamma - 1$ such users. The VCG auction of {\PRM} assigns the users of $\bigcup_{m \in M} \hat{P}_m$ to the slots of $B_s$ at the lowest locations, and thus, only slots at locations $1$ to $\bar{\tau} - 3\gamma - 1$ have a potential to be assigned a user. Since we assume $\ell(s) \geq \tau - 3\gamma$, no slot of $a$ is among these slots, and thus, $a$ is assigned no users by {\PRM}.
\end{proof}

The following lemma considers strategies of $a$ with a relatively small $\ell(s)$.
\begin{lemma} \label{lem:stable_length}
Given a strategy $s$ for advertiser $a$ such that $\ell(s) \leq \bar{\tau} - 2\gamma$ and a mediator $m \in M$, $|S_c(P \setminus P(m), B_s)| = |S_c(P \setminus P(m), B_{\bar{s}})|$.
\end{lemma}
\begin{proof}
The order of the users is identical in both $S_c(P \setminus P(m), B_s)$ and $S_c(P \setminus P(m), B_{\bar{s}})$. On the other hand, notice that $\ell(s)$ must be at least $1$, and thus, $\bar{\tau} \geq 2\gamma + 1$ by the  assumption $\ell(s) \leq \bar{\tau} - 2\gamma$. Moreover, this assumption implies that all the slots of $a$ appear at locations between $1$ and $\max\{\ell(s) + \gamma - 1, \gamma\} \leq \max\{\bar{\tau} - \gamma - 1, \gamma\} = \bar{\tau} - \gamma - 1$ in both $S_c(P \setminus P(m), B_s)$ and $S_c(P \setminus P(m), B_{\bar{s}})$. Hence, the sequences of the slots appearing in $S_c(P \setminus P(m), B_s)$ and $S_c(P \setminus P(m), B_{\bar{s}})$ starting from location $\bar{\tau} - \gamma$ are identical. This implies that for every $\bar{\tau} - \gamma \leq i \leq \min\{|P \setminus P(m)|, |B|\}$, $S_c(P \setminus P(m), B_s)$ assigns its user at location $i$ to its slot at location $i$ if and only if $S_c(P \setminus P(m), B_{\bar{s}})$ assigns its user at location $i$ to its slot at location $i$.

The above discussion implies that $|S_c(P \setminus P(m), B_s)| = |S_c(P \setminus P(m), B_{\bar{s}})|$ whenever it holds that:
\begin{equation} \label{eq:inequality_to_prove}
	|S_c(P \setminus P(m), B_{\bar{s}})| \geq \bar{\tau} - \gamma
	\enspace.
\end{equation}
Thus, the rest of the proof concentrate on proving Inequality~\eqref{eq:inequality_to_prove}. Let $p_{\bar{\tau}}$ and $s_{\bar{\tau}}$ be the user and slot at location $\bar{\tau}$ of $S_c(P, B_{\bar{s}})$. Since $\bar{\tau} = |S_c(P, B_{\bar{s}})|$ by definition, $p_{\bar{\tau}}$ is assigned to $s_{\bar{\tau}}$ by $S_c(P, B_{\bar{s}})$, and thus, $v(s_{\bar{\tau}}) > c(p_{\bar{\tau}})$. Next, let $p'_{\bar{\tau} - \gamma}$ and $s'_{\bar{\tau} - \gamma}$ be the user and slot at location $\bar{\tau} - \gamma$ of $S_c(P \setminus P(m), B_{\bar{s}})$. Clearly, $p'_{\bar{\tau} - \gamma}$ appears either at location $\bar{\tau}$ or at a smaller location in $S_c(P, B_{\bar{s}})$, and thus, $c(p'_{\bar{\tau} - \gamma}) \leq c(p_{\bar{\tau}})$. Moreover, $s'_{\bar{\tau} - \gamma}$ appears at location $\bar{\tau} - \gamma$ also in $S_c(P, B_{\bar{s}})$, and thus, $v(s'_{\bar{\tau} - \gamma}) > v(s_{\bar{\tau}})$. Combining all the above inequalities gives:
\[
	v(s'_{\bar{\tau} - \gamma})
	>
	v(s_{\bar{\tau}})
	>
	c(p_{\bar{\tau}})
	\geq
	c(p'_{\bar{\tau} - \gamma})
	\enspace.
\]
Hence, $p'_{\bar{\tau} - \gamma}$ is assigned to $s'_{\bar{\tau} - \gamma}$ by the canonical assignment $S_c(P \setminus P(m), B_{\bar{s}})$, which implies Inequality~\eqref{eq:inequality_to_prove}.
\end{proof}

\begin{corollary} \label{cor:same_prices}
For every mediator $m \in M$, $c_m$ is identical given that advertiser $a$ uses any strategy $s$ obeying $\ell(s) \leq \bar{\tau} - 2\gamma$.
\end{corollary}
\begin{proof}
Consider two strategies $s_1$ and $s_2$ for $a$ such that $\ell(s_1), \ell(s_2) \leq \bar{\tau} - 2\gamma$. Given that $a$ uses strategy $s_1$, {\PRM} selects $c_m$ as the cost of the user at location $|S_c(P \setminus P(m), B_{s_1})| - 4\gamma$ in $S_c(P \setminus P(m), B_{s_1})$ or $-\infty$ if $|S_c(P \setminus P(m), B_{s_1})| \leq 4\gamma$. Similarly, given that that $a$ uses strategy $s_2$, {\PRM} selects $c_m$ as the cost of the user at location $|S_c(P \setminus P(m), B_{s_2})| - 4\gamma$ in $S_c(P \setminus P(m), B_{s_2})$ or $-\infty$ if $|S_c(P \setminus P(m), B_{s_2})| \leq 4\gamma$. We claim that in both cases the same $c_m$ is produced by {\PRM}.

To see why this is the case observe that the two canonical assignments $S_c(P \setminus P(m), B_{s_1})$ and $S_c(P \setminus P(m), B_{s_2})$ share the order of the users, \ie, at a given location they both have the same user. Moreover, by Lemma~\ref{lem:stable_length}:
\[
	|S_c(P \setminus P(m), B_{s_1})|
	=
	|S_c(P \setminus P(m), B_{\bar{s}})|
	=
	|S_c(P \setminus P(m), B_{s_2})|
	\enspace.
	\qedhere
\]
\end{proof}

We are now ready to prove that truthfulness is a dominant strategy for $a$.

\begin{lemma} \label{lem:truthful_for_fixed_a}
Given {\PRM}, truthfulness is a dominant strategy for advertiser $a$.
\end{lemma}
\begin{proof}
If $\bar{\tau} \leq 3\gamma$, then, by Lemma~\ref{lem:no_users_case}, {\PRM} does not assign any users regardless of the strategy used by $a$, and thus, any strategy is a dominant strategy for $a$. Hence, we may assume $\bar{\tau} > 3\gamma$ for the remaining part of the proof. Let $s_t$ be the truthful strategy of $a$, and let $\hat{s}$ be any strategy such that $\bar{\tau} - 3\gamma \leq \ell(\hat{s}) \leq \bar{\tau} - 2\gamma$ (such a strategy exists by Observation~\ref{obs:density}). There are now two cases to consider, depending on the relationship between $\ell(s_t)$ and $\ell(\hat{s})$.

Consider first the case $\ell(s_t) \leq \ell(\hat{s}) \leq \bar{\tau} - 2\gamma$. By Corollary~\ref{cor:same_prices}, every strategy $s$ for $a$ having $\ell(s) \leq \bar{\tau} - 2\gamma$ results in the same values $c_m$ for all mediators, and thus, {\PRM} ends up using exactly the same VCG auction for all these strategies (up to the bid of the bidder corresponding to $a$). Since VCG auctions are incentive compatible, this means that reporting the truth is a best response for $a$ among all the possible strategies $s$ for $a$ obeying $\ell(s) \leq \bar{\tau} - 2\gamma$. Moreover, truthfulness guarantees a non-negative utility for $a$ since {\PRM} is advertiser-side IR by Observation~\ref{obs:advertisers_IR}. Thus, being truthful dominates also strategies $s$ for $a$ having $\ell(s) \geq \bar{\tau} - 2\gamma \geq \bar{\tau} - 3\gamma$ since such strategies result in $0$ utility for $a$ by Lemma~\ref{lem:no_users_case}.

Consider now the case $\ell(s_t) > \ell(\hat{s}) \geq \bar{\tau} - 3\gamma$. By Lemma~\ref{lem:no_users_case} $a$ is assigned no users (and thus, gets $0$ utility) when she uses her truthful strategy $s_t$. Let us observe what happens if advertiser $a$ is modified so that her new truthful strategy is $\hat{s}$. By Lemma~\ref{lem:no_users_case} the modified $a$ also gets $0$ utility when using her truthful strategy $\hat{s}$. On the other hand, the value of the modified advertiser for each user is at least as high as the value of the original advertiser, and thus, the modified advertiser gets at least as high utility as the original advertiser for every given outcome. Hence, to show that no strategy can give the original advertiser a positive utility, it is enough to show that the strategy maximizing the utility of the modified advertiser is her truthful strategy $\hat{s}$; which follows immediately from the first case (since the truthful strategy for the modified advertiser is $\hat{s}$, and $\hat{s}$ obeys $\ell(\hat{s}) \leq \bar{\tau} - 2\gamma$).
\end{proof}

\begin{corollary}
Given {\PRM}, truthfulness is a dominant strategy for every advertiser.
\end{corollary}
\begin{proof}
Follows from Lemma~\ref{lem:truthful_for_fixed_a} since $a$ is just an arbitrary fixed advertiser.
\end{proof}

%% file: SecRandomized.tex
\section{Randomized Mechanism} \label{sec:randomized}

In this section we describe the randomized mechanism ``Threshold by Partition Mechanism'' (\TPM) for our model. Unlike the mechanism {\PRM} from Section~\ref{sec:deterministic}, {\TPM} need not assume public knowledge about the advertisers' capacities, \ie, the advertisers now have multi-dimensional strategy spaces. On the other hand, {\TPM} assumes access to a value $\alpha \in [|S_c(P, B)|^{-1}, 1]$ such that we are guaranteed that:
\[
	\frac{u(a)}{|S_c(P, B)|} \leq	\alpha \quad \forall\; a \in A
	\qquad \text{and}	\qquad
	\frac{|P(m)|}{|S_c(P, B)|} \leq	\alpha \quad \forall\; m \in M
	\enspace.
\]
In other words, $\alpha$ is an upper bound on how large can be the capacity of an advertiser or the number of users of a mediator compared to the size of the optimal assignment $S_c(P, B)$. We remind the reader that $\alpha$ is related to the value $\gamma$ from Section~\ref{sec:deterministic} by the equation $\alpha = \gamma/\tau$, and thus, $\alpha$, like $\gamma$, can be informally understood as a bound on the importance of every single advertiser or mediator. It is important to note that $\alpha$ is well-defined only when $|S_c(P, B)| > 0$, and thus, we assume this inequality is true throughout the rest of the section.

A description of {\TPM} is given as Mechanism~\ref{mch:TPM}. Notice that Mechanism~\ref{mch:TPM} often refers to parameters of the model that are not known to the mechanism, such as the value of an advertiser or the number of users of a mediator. Whenever this happens, this should be understood as referring to the reported values of these parameters.

\vspace{2mm}
\noindent \begin{minipage}{\textwidth}
\captionsetup{type=mechanism}
\noindent \rule{\linewidth}{0.8pt}
\vspace{-6mm}\captionof{mechanism}{Threshold by Partition Mechanism (\TPM)}\label{mch:TPM}
\noindent \rule{\linewidth}{0.8pt}
\vspace{-5.5mm}
\begin{compactenum}[\bfseries 1.]
	\item Let $M_L$ be a set of mediators containing each mediator $m \in M$ with probability $\min\{17\sqrt[3]{\alpha}, 1\}$, independently. Similarly, $A_L$ is a set of advertisers containing each advertiser $a \in A$ with probability $\min\{17\sqrt[3]{\alpha}, 1\}$, independently. Intuitively, the subscript $L$ in $M_L$ and $A_L$ stands for ``low priority''.
\end{compactenum}
\end{minipage}
{
\hphantom{1}\setlength{\plitemsep}{3pt}
\begin{compactenum}[\bfseries 1.]
	\setcounter{enumi}{1}
	\item Let $\sigma_A$ be an arbitrary order over the advertisers that places the advertisers of $A_L$ after all the other advertisers and is independent of the reports received by the mechanism. Similarly, $\sigma_M$ is an arbitrary order over the mediators that places the mediators of $M_L$ after all the other mediators and is independent of the reports received by the mechanism.
	\item Partition the mediators of $M$ into two sets $M_1$ and $M_2$ by adding each mediator $m \in M$ with probability $1/2$, independently, to $M_1$ and otherwise to $M_2$. Similarly, partition the advertisers of $A$ into two sets $A_1$ and $A_2$ by adding each advertiser $a \in A$ with probability $1/2$, independently, to $A_1$ and otherwise to $A_2$. The rest of the algorithm explains how to assign users of mediators from $M_1$ to slots of advertisers from $A_1$, and how to charge advertisers of $A_1$ and pay mediators of $M_1$. Analogous steps, which we omit, should be added for handling the advertisers of $A_2$ and the mediators of $M_2$.
	\item Let $\hat{p}$ and $\hat{b}$ be the user and slot at location $\lceil (1 - 4\sqrt[3]{\alpha}) \cdot |S_c(P(M_2), B(A_2))| \rceil$ of the canonical solution $S_c(P(M_2), B(A_2))$. If $(1 - 4\sqrt[3]{\alpha}) \cdot |S_c(P(M_2), B(A_2))| \leq 0$, then the previous definition of $\hat{p}$ and $\hat{b}$ cannot be used. Instead define $\hat{p}$ as a dummy user of cost $-\infty$ and $\hat{b}$ as a dummy slot of value $\infty$. Using $\hat{p}$ and $\hat{b}$ define now two sets
	\[
		\hat{P} = \{p \in P(M_1) \mid c(p) < c(\hat{p})\}
		\qquad \text{and} \qquad
		\hat{B} = \{b \in B(A_1) \mid v(b) > v(\hat{b})\}
		\enspace.
	\]
	It is important to note that $\hat{P}$ and $\hat{B}$ are empty whenever $\hat{p}$ and $\hat{b}$ are dummy user and slot, respectively.
\end{compactenum}
}
\vspace{2mm}
\noindent\begin{minipage}{\textwidth}
\begin{compactenum}[\bfseries 1.]
	\setcounter{enumi}{4}
	\item While there are unassigned users in $\hat{P}$ and unassigned slots in $\hat{B}$ do the following:
	\begin{compactitem}[$\bullet$]
		\item Let $m$ be the earliest mediator in $\sigma_M$ having unassigned users in $\hat{P}$.
		\item Let $a$ be the earliest advertiser in $\sigma_A$ having unassigned slots in $\hat{B}$.
		\item Assign the unassigned user of $\hat{P} \cap P(m)$ with the lowest cost to an arbitrary unassigned slot of $\hat{B} \cap B(a)$, charge a payment of $v(\hat{b})$ to advertiser $a$ and transfer a payment of $c(\hat{p})$ to mediator $m$.\footnotemark
	\end{compactitem}
\end{compactenum}
\vspace{-3mm}\noindent \rule{\linewidth}{0.8pt}
\end{minipage}
%\vspace{2mm}
%\noindent\begin{minipage}{\textwidth}
%\begin{compactenum}[\bfseries 1.]
%	\setcounter{enumi}{5}
%	\item For every mediator $m \in M_1$:
%	\begin{compactitem}[$\bullet$]
%		\item If all the users of $\hat{P} \cap P(m)$ are assigned to slots, recommend to $m$ to forward a payment of $c(\hat{p})$ to each one of his assigned users.
%		\item Otherwise, let $p$ be the unassigned user of $\hat{P} \cap P(m)$ with the minimum cost. Recommend to $m$ to forward a payment of $c(p)$ to each one of his assigned users.\footnotemark
%	\end{compactitem}
%\end{compactenum}
%vspace{-3mm}\noindent \rule{\linewidth}{0.8pt}
%\end{minipage}
\footnotetext{Note that $m$ is paid $c(\hat{p})$ for the assignment of each one of his users. Hence, $m$ is always BB since the membership of $p$ in $\hat{P}$ implies $c(p) < c(\hat{p})$ (and $c(p) \leq c(\hat{p})$ when the costs are compared as numbers).}
\vspace{2mm}

Let us recall our result regarding {\TPM}\inConference{ (the proof of Theorem~\ref{thm:TPM} is deferred to Appendix~\ref{app:randomized})}.

\begin{reptheorem}{thm:TPM}
{\TPM} is BB, IR, IC and $(1 - 28\sqrt[3]{\alpha} - 20e^{-2/\sqrt[3]{\alpha}})$-competitive.
\end{reptheorem}

Intuitively, the analysis of {\TPM} works by exploiting concentration results showing that the canonical assignments $S_c(P(M_1), B(A_1))$ and $S_c(P(M_2), B(A_2))$ are \inArXiv{not very different from each other}\inConference{quite similar}. This similarity allows us to use information from $S_c(P(M_2), B(A_2))$ to set the payments charged to advertisers of $B(A_1)$ and payed to mediators of $P(M_1)$, and vice versa, while keeping a reasonable competitive ratio. The advantage of setting the payments this way is that it reduces the control players have on the payments they have to pay or are paid, which helps the mechanism to be IC.

\inArXiv{\input{AppRandomized}}

%% file: AppRandomized.tex
\inConference{\section{Proof of Theorem~\ref{thm:TPM}} \label{app:randomized}}

\inArXiv{In the rest of}\inConference{In} this section we prove Theorem~\ref{thm:TPM}. We start with the following observation.

\begin{observation}
{\TPM} is BB.
\end{observation}
\begin{proof}
We show that whenever {\TPM} assigns a user $p$ to a slot $b$, it charges the advertiser of $b$ more than it pays the mediator of $p$. Consider an arbitrary ordered pair $(p, b)$ from the assignment produced by {\TPM}. We assume without loss of generality that $p \in P(M_1)$; the other case is symmetric. The fact that $p$ is assigned implies $p \in \hat{P}$, hence, $\hat{P}$ is non-empty. Similarly, the fact that a user is assigned to $b$ implies that $\hat{B}$ is not empty.

Recall that the fact that $\hat{P}$ and $\hat{B}$ are not empty implies that $\hat{p}$ and $\hat{b}$ are not dummy user and slot, respectively. This means that $\hat{p}$ and $\hat{b}$ are matched by the canonical assignment $S_c(P(M_2), B(A_2))$. Since a canonical assignment never assigns a user $p'$ to a slot $b'$ when $c(p') > v(b')$, we get $c(\hat{p}) < v(\hat{b})$. The proof now completes by observing that the advertiser of $b$ is charged $v(\hat{b})$ for the assignment of $p$ to $b$, while the mediator of $p$ is paid $c(\hat{p})$ for this assignment.
\end{proof}

\subsection{The Incentive Properties of {\TPM}}

In this section we prove the incentive parts of Theorem~\ref{thm:TPM}. Specifically, we prove that {\TPM} is IR and IC.% The next lemma analyzes the incentive properties of {\TPM} for mediators.

\begin{lemma}
For every mediator $m$ (advertiser $a$), {\TPM} is IR for $m$ ($a$), and truthfulness is a dominant strategy for him (her).
\end{lemma}
\begin{proof}
We prove the lemma only for mediators. The proof for advertisers is analogous (with slots exchanging roles with users, $v(\hat{b})$ exchanging roles with $c(\hat{p})$, etc.), and thus, we omit it. Additionally, we assume without loss of generality that $m \in M_1$ (the other case is symmetric).

Recall that {\TPM} calculates a set $\hat{P}$ of users, and each user $p' \in P(M_1)$ belongs to $\hat{P}$ if and only if her cost is lower than some threshold $c(\hat{p})$. Additionally, note that {\TPM} calculates the threshold $c(\hat{p})$ based on the reports of advertisers and mediators in $A_2$ and $M_2$, respectively. Thus, $m$, which belongs to $M_1$, cannot affect this threshold. Similarly, $m$ cannot affect the set of slots $\hat{B}$.

Whenever a user $p \in P(m)$ is assigned to a slot the utility of $m$ decreases by $c(p)$ and increases by the payment $m$ gets, which is $c(\hat{p})$. In other words, the utility of $m$ changes by $c(\hat{p}) - c(p)$. When $m$ is truthful this change is always non-negative since $p$ can be assigned only when she belongs to $\hat{P}$, which implies that the cost of $p$ is smaller than $\hat{p}$. This already proves that the utility of $m$ is non-negative when he is truthful, and thus, {\TPM} is IR for $m$.

We claim that there exists a value $k$ which is independent of the report of $m$ such that for any report of $m$ the mechanism assigns the $\min\{k, |\hat{P} \cap P(m)|\}$ users of $m$ with the lowest reported costs. Before proving this claim, let us explain why the lemma follows from this claim. The above description shows that the utility of $m$ changes by a $c(\hat{p}) - c(p)$ for every assigned user $p \in P(m)$, thus, $m$ wishes to assign as many as possible users having cost less than $c(\hat{p})$, and if he cannot assign all of them then he prefers to assign the users with the lowest costs. By reporting truthfully $m$ guarantees that only users of cost less than $c(\hat{p})$ get to $\hat{P}$, and thus, has a chance to be assigned. Moreover, by the above claim the mechanism assigns the $k$ users of $\hat{P} \cap P(m)$ with the lowest costs (or all of them if $|\hat{P} \cap P(m)| < k$). Hence, the above claim indeed implies that truthfulness is a dominant strategy for $m$.

We are only left to prove the above claim. One can view {\TPM} as considering the mediators according to the order $\sigma_M$. For every mediator of $m' \in M_1$ {\TPM} assigns the users of $\hat{P} \cap P(m')$ one by one (in an increasing reported costs order). Every assignment of a user ``consumes'' one unassigned slot of $\hat{B}$, and the assignment of users stops when all the users of $\hat{P}$ are assigned, or there are no more unassigned slots in $\hat{B}$. This means that if {\TPM} assigns users to all the slots of $\hat{B}$ before considering $m$, then no user of $m$ is assigned. The claim is true in this case with $k = 0$. Otherwise, we choose $k$ to be the number of unassigned $\hat{B}$ slots immediately before {\TPM} considers $m$. Notice that the report of $m$ does not affect the behavior of {\TPM} up to the point it considers $m$, and thus, $k$ is independent of the report of $m$. If there are more than $k$ users in $\hat{P} \cap P(m)$, then only the $k$ of them with the lowest costs are assigned before {\TPM} consumes all the unassigned slots of $\hat{B}$ and stops. Otherwise, if $|\hat{P} \cap P(m)| \leq k$ then {\TPM} manages to assigns all the users of $\hat{P} \cap P(m)$ before it runs out of unassigned slots of $\hat{B}$.
\end{proof}

%The the next lemma considers the incentive properties of {\TPM} for advertisers. The proof of this lemma is analogous to the proof of the previous lemma (with slots exchanging roles with users, $v(\hat{b})$ exchanging roles with $c(\hat{p})$, etc.), and thus, we omit it.
%\begin{lemma}
%For every advertiser $a$, {\TPM} is IR for $a$, and truthfulness is a dominant strategy for her.
%\end{lemma}

\subsection{The Competitive Ratio of {\TPM}}

In this section we analyze the competitive ratio of {\TPM}. %Recall that $\tau$ was defined as a shorthand for $|S_c(P, B)|$.
Let us define $\tilde{P}$ ($\tilde{B}$) as the set of the users (slots) at locations $1$ to $\lceil (1 - 11\sqrt[3]{\alpha})\tau \rceil$ of the canonical assignment $S_c(P, B)$ ($\tilde{P}$ and $\tilde{B}$ are defined as the empty set when $1 - 11\sqrt[3]{\alpha} \leq 0$). The following observation shows that most of the gain from trade of the canonical assignment $S_c(P, B)$ comes from the users and slots of $\tilde{P}$ and $\tilde{B}$, respectively. For convenience, let us denote by $P_o$ the set of users that are assigned by $S_c(P, B)$, and by $B_o$ the set of slots that are assigned some user by $S_c(P, B)$.

\begin{observation} \label{obs:tilde_elements}
$\sum_{b \in \tilde{B}} v(b) - \sum_{p \in \tilde{P}} c(p) \geq (1 - 11\sqrt[3]{\alpha}) \cdot \GfT(S_c(P, B))$.
\end{observation}
\begin{proof}
If $1 - 11\sqrt[3]{\alpha} \leq 0$, then both $\tilde{B}$ and $\tilde{P}$ are empty, and the inequality that we need to prove holds since its left hand side is $0$ and its right hand side is non-positive (recall that $S_c(P, B)$ is an assignment of users from $P$ to slots of $B$ maximizing the gain from trade, and thus, its gain from trade is at least $0$ since $\GfT(\varnothing) = 0$). Thus, we may assume in the rest of the proof that $1 - 11\sqrt[3]{\alpha} > 0$.

Since $\tilde{B}$ contains the $\lceil (1 - 11\sqrt[3]{\alpha})\tau \rceil$ slots with the largest values among the slots of $B_o$, we get:
\[
	\sum_{b \in \tilde{B}} v(b)
	\geq
	\lceil (1 - 11\sqrt[3]{\alpha})\tau \rceil \cdot \frac{\sum_{b \in B_o} v(b)}{\tau}
	\enspace.
\]
Similarly, since $\tilde{P}$ contains the $\lceil (1 - 11\sqrt[3]{\alpha})\tau \rceil$ users with the lowest costs among the users of $P_o$, we get:
\[
	\sum_{p \in \tilde{A}} c(p)
	\leq
	\lceil (1 - 11\sqrt[3]{\alpha})\tau \rceil \cdot \frac{\sum_{c \in P_o} c(p)}{\tau}
	\enspace.
\]
Combining the two inequities gives:
\begin{align*}
	\sum_{b \in \tilde{B}} v(b) - \sum_{p \in \tilde{P}} c(p)
	\geq{} &
	\lceil (1 - 11\sqrt[3]{\alpha})\tau \rceil \cdot \frac{\left(\sum_{b \in B_o} v(b) - \sum_{p \in P_o} c(p)\right)}{\tau}\\
	={} &
	\lceil (1 - 11\sqrt[3]{\alpha})\tau \rceil \cdot \frac{\GfT(S_c(P, B))}{\tau}
	\geq
	(1 - 11\sqrt[3]{\alpha}) \cdot \GfT(S_c(P, B))
	\enspace.
	\qedhere
\end{align*}
\end{proof}

Observation~\ref{obs:tilde_elements} shows that one can prove a competitive ratio for {\TPM} by relating the gain from trade of the assignment it produces to the gain from trade obtained by assigning the users of $\tilde{P}$ to the slots $\tilde{B}$. The following lemma is a key lemma we use to relate the two gains. 

\begin{lemma} \label{lem:event_summary}
There exists an event $\cE_1$ of probability at least $1 - 10e^{-2/\sqrt[3]{\alpha}}$ such that $\cE_1$ implies the following:
\begin{center}
\begin{tabular}{llll}
	(i)  & $\tilde{B} \cap B(A_1) \subseteq \hat{B}$ \hspace{3cm} & (iii) & $|\hat{P} \setminus P(M_L)| \leq |\hat{B}|$\\
	(ii) & $\tilde{P} \cap P(M_1) \subseteq \hat{P}$ \hspace{3cm} & (iv)  & $|\hat{B} \setminus B(A_L)| \leq |\hat{P}|$ \\
	(v)  & \multicolumn{3}{p{11cm}}{$c(p) \leq \ell(P, B) \leq v(b)$ for every user $p \in \hat{P}$ and slot $b \in \hat{B}$, where $\ell(P, B)$ is a value which is independent of the random coins of {\TPM} and obeys $c(p) \leq \ell(P, B) \leq v(b)$ for every $p \in P_o$ and $b \in B_o$.}
\end{tabular}
\end{center}
\end{lemma}

Before proving Lemma~\ref{lem:event_summary}, let us explain how the competitive ratio of {\TPM} follows from it. Let $\hat{S}$ be the assignment produced by {\TPM}.

\begin{lemma} \label{lem:joined_event}
There exists an event $\cE$ of probability at least $1 - 20e^{-2/\sqrt[3]{\alpha}}$ such that $\cE$ implies:
\[
	\GfT(\hat{S})
	\geq
	\sum_{b \in \tilde{B} \setminus B(A_L)} [v(b) - \ell(P, B)] + \sum_{p \in \tilde{P} \setminus P(M_L)} [\ell(P, B) - c(p)]
	\enspace.
\]
\end{lemma}
\begin{proof}
Assume first that the event $\cE_1$ holds. Lemma~\ref{lem:event_summary} shows that given $\cE_1$ we have $|\hat{P} \setminus P(M_L)| \leq |\hat{B}|$, hence, {\TPM} assigns at least $|\hat{P} \setminus P(M_L)|$ users. Additionally, since {\TPM} assigns all the users of $\hat{P} \setminus P(M_L)$ before it starts assigning users of $\hat{P} \cap P(M_L)$, we get that all the users of $\hat{P} \setminus P(M_L)$ are assigned by $\hat{S}$ given $\cE_1$. On the other hand, Lemma~\ref{lem:event_summary} also shows that given $\cE_1$ all the users of $\tilde{P} \cap P(M_1)$ belong to $\hat{P}$, and thus, the users of $(\tilde{P} \cap P(M_1)) \setminus P(M_L)$ are all assigned by $\hat{S}$. A similar argument shows that the slots of $(\tilde{B} \cap B(A_1)) \setminus B(A_L)$ are all assigned users by $\hat{S}$ given $\cE_1$. Finally, observe that $\cE_1$ also implies that $c(p) \leq \ell(P, B) \leq v(b)$ for every pair $(p, b) \in \hat{S} \cap (P(M_1) \times B(A_1))$.

{\TPM} assigns users of $P(M_2)$ to slots of $B(A_2)$ in an analogous way to the assignment of users of $P(M_1)$ to slots of $B(A_1)$. Thus, by symmetry, there must exist an event $\cE_2$ of probability at least $1 - 10e^{-2/\sqrt[3]{\alpha}}$ (the probability of $\cE_1$) that guarantees the following three things.
\begin{compactitem}
	\item The users of $(\tilde{P} \cap P(M_2)) \setminus P(M_L)$ are all assigned by $\hat{S}$.
	\item The slots of $(\tilde{B} \cap B(A_2)) \setminus B(A_L)$ are all assigned users by $\hat{S}$.
	\item $c(p) \leq \ell(P, B) \leq v(b)$ for every pair $(p, b) \in \hat{S} \cap (P(M_2) \times B(A_2))$.
\end{compactitem}

Let us now define $\cE$ as the event that both $\cE_1$ and $\cE_2$ happen together. By the union bound the probability of $\cE$ is at least $1 - 20e^{-2/\sqrt[3]{\alpha}}$ as promised. Additionally, combining the properties that follow from $\cE_1$ and $\cE_2$ by the above discussion, we get that $\cE$ implies the following properties.
\begin{compactitem}
	\item The users of $\tilde{P} \setminus P(M_L)$ are all assigned by $\hat{S}$.
	\item The slots of $\tilde{B} \setminus B(A_L)$ are all assigned users by $\hat{S}$.
	\item $c(p) \leq \ell(P, B) \leq v(b)$ for every pair $(p, b) \in \hat{S}$.
\end{compactitem}
The last property holds since {\TPM} assigns users of $P(M_1)$ ($P(M_2)$) only to slots of $B(A_1)$ ($B(A_2)$), and thus, $(\hat{S} \cap (P(M_2) \times B(A_2))) \cup (\hat{S} \cap (P(M_2) \times B(A_2))) = \hat{S}$.

In the rest of the proof we assume that $\cE$ happens. Consider an ordered pair $(p, b) \in \hat{S}$. The contribution of $(p, b)$ to $\GfT(\hat{S})$ is:
\[
	v(b) - c(p)
	=
	[v(b) - \ell(P, B)] + [\ell(P, B) - c(p)]
	\enspace.
\]
Notice that the two terms that appear in brackets on the right hand side of the last equation are both positive given $\cE$. This allows us to lower bound the gain from trade of $\hat{S}$ as follows:
\begin{align*}
	\GfT(\hat{S})
	={} &
	\sum_{(p, b) \in \hat{S}} [v(b) - c(p)]
	=
	\sum_{(p, b) \in \hat{S}} \{[v(b) - \ell(P, B)] + [\ell(P, B) - c(p)]\}\\
	\geq{} &
	\sum_{b \in \tilde{B} \setminus B(A_L)} [v(b) - \ell(P, B)] + \sum_{p \in \tilde{P} \setminus P(M_L)} [\ell(P, B) - c(p)]
	\enspace.
	\qedhere
\end{align*}
\end{proof}

\begin{corollary}
{\TPM} is at least $(1 - 28\sqrt[3]{\alpha} - 20e^{-2/\sqrt[3]{\alpha}})$-competitive.
\end{corollary}
\begin{proof}
The corollary is trivial when $28\sqrt[3]{\alpha} + 20e^{-2/\sqrt[3]{\alpha}} > 1$. Thus, we assume in this proof $28\sqrt[3]{\alpha} + 20e^{-2/\sqrt[3]{\alpha}} \leq 1$. Let $\Val(M_L, A_L)$ denote the expression:
\[
	\sum_{b \in \tilde{B} \setminus B(A_L)} [v(b) - \ell(P, B)] + \sum_{p \in \tilde{P} \setminus P(M_L)} [\ell(P, B) - c(p)]
	\enspace.
\]
Since every slot belongs to $B(A_L)$ with probability $17\sqrt[3]{\alpha}$ and every user belongs to $P(M_L)$ with the same probability, we get:
\begin{align*}
	\bE[\Val(M_L, A_L)]
	={} &
	(1 - 17\sqrt[3]{\alpha}) \cdot \sum_{b \in \tilde{B}} [v(b) - \ell(P, B)] + (1 - 17\sqrt[3]{\alpha}) \cdot \sum_{p \in \tilde{P}} [\ell(P, B) - c(b)]\\
	={} &
	(1 - 17\sqrt[3]{\alpha}) \cdot \Val(\varnothing, \varnothing)
	\enspace.
\end{align*}
Additionally, the definition of $\ell(P, B)$ guarantees that $v(b) - \ell(P, B) \geq 0$ and $\ell(P, B) - c(p) \geq 0$ for every $b \in \tilde{B} \subseteq B_o$ and $p \in \tilde{P} \subseteq P_o$. Thus, $\Val(M_L, A_L) \leq \Val(\varnothing, \varnothing)$ for every two sets $M_L \subseteq M$ and $A_L \subseteq A$. Using Lemma~\ref{lem:joined_event} and the observation that {\TPM} always produces assignments of non-negative gain from trade, we now get:
\begin{align} \label{eq:gain_from_trade_hat_S}
	\bE[\GfT(\hat{S})]&
	=
	\Pr[\cE] \cdot \bE[\GfT(\hat{S}) \mid \cE] + \Pr[\neg \cE] \cdot \bE[\GfT(\hat{S}) \mid \neg \cE]\nonumber\\
	\geq{} &
	\Pr[\cE] \cdot \bE[\Val(M_L, A_L) \mid \cE]
	=
	\bE[\Val(M_L, A_L)] - \Pr[\neg \cE] \cdot \bE[\Val(M_L, A_L) \mid \neg \cE]\nonumber\\
	\geq{} &
	(1 - 17\sqrt[3]{\alpha}) \cdot \Val(\varnothing, \varnothing) - \Pr[\neg \cE] \cdot \Val(\varnothing, \varnothing)
	=
	(1 - 17\sqrt[3]{\alpha} - \Pr[\neg \cE]) \cdot \Val(\varnothing, \varnothing)
	\enspace.
\end{align}

Recall that $\Pr[\neg \cE] \leq 20e^{-2/\sqrt[3]{\alpha}}$ by Lemma~\ref{lem:joined_event}. Additionally, Observation~\ref{obs:tilde_elements} and the fact that $|\tilde{P}| = |\tilde{B}|$ by definition imply together:
\begin{align*}
	\Val(\varnothing, \varnothing)
	={} &
	\sum_{b \in \tilde{B}} [v(b) - \ell(P, B)] + \sum_{p \in \tilde{P}} [\ell(P, B) - c(p)]\\
	={} &
	\sum_{b \in \tilde{B}} v(b) - \sum_{p \in \tilde{P}} c(p)
	\geq
	(1 - 11\sqrt[3]{\alpha}) \cdot \GfT(S_c(P, A))
	\enspace.
\end{align*}
Plugging the last observations into Inequality~\ref{eq:gain_from_trade_hat_S} gives:
\begin{align*}
	\bE[\GfT(\hat{S})]
	\geq{} &
	(1 - 17\sqrt[3]{\alpha} - \Pr[\neg \cE]) \cdot \Val(\varnothing, \varnothing)\\
	\geq{} &
	(1 - 17\sqrt[3]{\alpha} - 20e^{-2/\sqrt[3]{\alpha}}) \cdot (1 - 11\sqrt[3]{\alpha}) \cdot \GfT(S_c(P, A))\\
	\geq{} &
	(1 - 28\sqrt[3]{\alpha} - 20e^{-2/\sqrt[3]{\alpha}}) \cdot \GfT(S_c(P, B))
	\enspace.
\end{align*}
The corollary now follows by recalling that $S_c(P, B)$ is the assignment of users from $P$ to slots of $B$ which maximizes the gain from trade.
\end{proof}

It remains now to prove Lemma~\ref{lem:event_summary}. Let us begin with the following technical lemma. 
\begin{lemma} \label{lem:length_concentration}
Given a subset $B' \subseteq B_o$ and a probability $q \in [0, 1]$, let $B'[q]$ be a random subset of $B'$ constructed as follows: for every advertiser $a \in A$, independently, with probability $q$ the slots of advertiser $a$ that belong to $B'$ appear also in $B'[q]$. Then, for every $\beta \in (0, 1]$:
\[
	\Pr[||B'[q]| - q \cdot |B'|| \geq \beta \tau]
	\leq
	2e^{-2\beta^2/\alpha}
	\enspace.
\]
Similarly, given a subset $P' \subseteq P_o$ and a probability $q \in [0, 1]$, let $P'[q]$ be a random subset of $P'$ constructed as follows: for every mediator $m \in M$, independently, with probability $q$ the users of mediator $m$ that belong to $P'$ appear also in $P'[q]$. Then, for every $\beta \in (0, 1]$:
\[
	\Pr[||P'[q]| - q \cdot |P'|| \geq \beta \tau]
	\leq
	2e^{-2\beta^2/\alpha}
	\enspace.
\]
\end{lemma}
\begin{proof}
We prove the first inequality; the second inequality is analogous. First, observe that the lemma is trivial when $B' = \varnothing$ since $B' = \varnothing$ implies $||B'[q]| - q \cdot |B'|| = 0 < \beta \tau$. Thus, we may assume in the rest of the proof $B' \neq \varnothing$. For every advertiser $a \in A$, let $X_a$ be an indicator for the event that slots of $a$ appear in $B'[q]$. Then:
\[
	|B'[q]|
	=
	\sum_{a \in A} X_a \cdot |B' \cap B(a)|
	\enspace.
\]
The definition of $\alpha$ implies $|B(a)| \leq \alpha \tau$ for every advertiser $a \in A$, and thus, $0 \leq |B' \cap B(a)| \leq \alpha \tau$. Hence, by Hoeffding's inequality:
\begin{align*}
	\Pr[||B'[q]| - q \cdot |B'|| \geq \beta \tau]
	={} &
	\Pr[||B'[q]| - \bE[|B'[q]|]| \geq \beta \tau]
	\leq
	2e^{-\frac{2(\beta \tau)^2}{\sum_{a \in A} |B' \cap B(a)|^2}}\\
	\leq{} &
	2e^{-\frac{2(\beta \tau)^2}{\alpha \tau \cdot \sum_{a \in A} |B' \cap B(a)|}}
	=
	2e^{-\frac{2\beta^2 \tau}{\alpha \cdot |B'|}}
	\leq
	2e^{-\frac{2\beta^2 \tau}{\alpha \cdot |B_o|}}
	=
	2e^{-\frac{2\beta^2}{\alpha}}
	\enspace.
	\qedhere
\end{align*}
\end{proof}

\hypertargettop{event:E_prime}Let $\cE'$ be the event that the following inequalities are all true (at the same time):
\begin{center}
\begin{tabular}{llll}
	(i)  & $||B_o \cap B(A_2)| - |B_o|/2| \leq \sqrt[3]{\alpha} \cdot \tau$ \hspace{1cm} & (iii) & $||\tilde{B} \cap B(A_2)| - |\tilde{B}|/2| \leq \sqrt[3]{\alpha} \cdot \tau$\\
	(ii) & $||P_o \cap P(M_2)| - |P_o|/2| \leq \sqrt[3]{\alpha} \cdot \tau$              & (iv)  & $||\tilde{P} \cap P(M_2)| - |\tilde{P}|/2| \leq \sqrt[3]{\alpha} \cdot \tau$
\end{tabular}
\end{center}

\begin{observation} \label{obs:E_prime_probability}
$\Pr[\cE'] \geq 1 - 8e^{-2/\sqrt[3]{\alpha}}$.
\end{observation}
\begin{proof}
Observe that $B_o \cap B(A_2)$, $\tilde{B} \cap B(A_2)$, $P_o \cap P(M_2)$ and $\tilde{P} \cap P(M_2)$ have the same distributions as $B_o[1/2]$, $\tilde{B}[1/2]$, $P_o[1/2]$ and $\tilde{P}[1/2]$, respectively. Moreover, by definition $\tilde{B} \subseteq B_o$ and $\tilde{P} \subseteq P_o$. Hence, by Lemma~\ref{lem:length_concentration}, each one of the four inequalities defining $\cE'$ holds with probability at least $	1 - 2e^{-2/\sqrt[3]{\alpha}}$. The observation now follows by the union bound.
\end{proof}

Next, we need the following useful observation.

\begin{observation} \label{obs:length_characterization}
It always holds that:
\[
	\min\{|P_o \cap P(M_2)|, |B_o \cap B(A_2)|\}
	\leq
	|S_c(P(M_2), B(A_2))|
	\leq
	\max \{|P_o \cap P(M_2)|, |B_o \cap B(A_2)|\}
	\enspace.
\]
\end{observation}
\begin{proof}
Let $p_\tau$ and $b_\tau$ be the user and slot at location $\tau$ of $S_c(P, B)$, respectively. The definition of a canonical assignment guarantees $c(p_{\tau}) < v(b_{\tau})$. Additionally, the slots of $B_o$ all appear in locations $1$ to $\tau$ of $S_c(P, B)$, and thus, they all have values at least as large as $v(b_\tau)$. Similarly, the users of $P_o$ all have costs at most as large as $c(p_\tau)$. Combining these observations, we get: $c(p) \leq c(p_\tau) < v(b_\tau) \leq v(b)$ for every $p \in P_o$ and $b \in B_o$.

The slots at locations $1$ to $|B_o \cap B(A_2)|$ of $S_c(P(M_2), B(A_2))$ all belong to $B_o$ since $B_o$ contains the $\tau$ slots with the largest values. Similarly, the users at locations $1$ to $|P_o \cap P(M_2)|$ belong to $P_o$. Combining both observations, we get that for every location $1 \leq i \leq \min\{|P_o \cap P(M_2)|, |B_o \cap B(A_2)|\}$, the user $p'_i$ at location $i$ of $S_c(P(M_2), B(A_2))$ and the slot $b'_i$ at this location belong to $P_o$ and $B_o$, respectively, and thus, $c(p'_i) < v(b'_i)$. Hence, by the definition a canonical assignment, the pair $(p'_i, b'_i)$ belongs to $S_c(P(M_2), B(A_2))$ for every $1 \leq i \leq \min\{|P_o \cap P(M_2)|, |B_o \cap B(A_2)|\}$; which completes the proof of the first inequality we need to prove.

Assume towards a contradiction that the second inequality we need to prove is wrong. In other words, we assume $|S_c(P(M_2), B(A_2))| > \max \{|P_o \cap P(M_2)|, |B_o \cap B(A_2)|\}$. Let $j = \max \{|P_o \cap P(M_2)|, |B_o \cap B(A_2)|\} + 1$, and let $p'_j$ and $b'_j$ be the user and slot at location $j$ of $S_c(P(M_2), B(A_2))$, respectively. Our assumption implies that $(p'_j , b'_j)$ belongs to $S_c(P(M_2), B(A_2))$, and thus, $c(p'_j) < v(b'_j)$. On the other hand, only the users at locations $1$ to $|P_o \cap P(M_2)|$ of $S_c(P(M_2), B(A_2))$ belong to $P_o$, hence, $p'_j$ does not belong to $P_o$. The user with the lowest cost that does not belong to $P_o$ is the user $p_{\tau + 1}$ at location $\tau + 1$ of $S_c(P, B)$. Thus, we get: $c(p'_j) \geq c(p_{\tau + 1})$. Similarly, we can also get $v(b'_j) \leq v(b_{\tau + 1})$, where $b_{\tau + 1}$ is the slot at location $\tau + 1$ of $S_c(P, B)$. Combining the above inequalities gives:
\[
	c(p_{\tau + 1})
	\leq
	c(p'_j)
	<
	v(p'_j)
	\leq
	v(b_{\tau + 1})
	\enspace,
\]
which contradicts the fact that $p_{\tau + 1}$ is not assigned to $b_{\tau + 1}$ by the canonical assignment $S_c(P, B)$.
\end{proof}

The next few claims use the last observation to prove a few properties that hold given $\cE$'. 

\begin{lemma} \label{lem:upper_inclusion}
The event \hyperlink{event:E_prime}{$\cE'$} implies: $\hat{P} \subseteq P_o$ and $\hat{B} \subseteq B_o$.
\end{lemma}
\begin{proof}
We prove the first inclusion. The other inclusion is analogous. Observation~\ref{obs:length_characterization} and the definition of $\cE'$ imply:
\begin{align*}
	|S_c(P(M_2), B(A_2))|
	\leq{} &
	\max \{|P_o \cap P(M_2)|, |B_o \cap B(A_2)|\}\\
	\leq{} &
	\max\{|P_o|/2 + \sqrt[3]{\alpha} \cdot \tau, |B_o|/2  + \sqrt[3]{\alpha} \cdot \tau\}
	=
	(1/2 + \sqrt[3]{\alpha})\tau
	\enspace.
\end{align*}
Using the definition of $\cE'$ again gives:
\begin{align*}
	|P_o \cap P(M_2)|
	\geq{} &
	|P_o|/2 - \sqrt[3]{\alpha} \cdot \tau
	=
	(1/2 - \sqrt[3]{\alpha})\tau\\
	\geq{} &
	(1 - 4\sqrt[3]{\alpha}) \cdot (1/2 + \sqrt[3]{\alpha})\tau
	\geq
	(1 - 4\sqrt[3]{\alpha}) \cdot |S_c(P(M_2), B(A_2))|
	\enspace,
\end{align*}
which implies, since $|P_o \cap P(M_2)|$ is integral,
\begin{equation} \label{eq:threshold_in_original}
	|P_o \cap P(M_2)|
	\geq
	\lceil (1 - 4\sqrt[3]{\alpha}) \cdot |S_c(P(M_2), B(A_2))| \rceil
	\enspace.
\end{equation}

If $\hat{p}$ is a dummy user then $\hat{P}$ is empty, which makes the claim $\hat{P} \subseteq P_o$ trivial. Thus, we may assume that $\hat{p}$ is the user at location $\lceil (1 - 4\sqrt[3]{\alpha}) \cdot |S_c(P(M_2), B(A_2))| \rceil$ of the canonical assignment $S_c(P(M_2), B(A_2))$. Hence, Inequality~\eqref{eq:threshold_in_original} and the observation that the users of $P_o \cap P(M_2)$ are the users with the lowest costs in $P(M_2)$ imply together that $\hat{p}$ belongs to the set $P_o \cap P(M_2) \subseteq P_o$. 
On the other hand, $P_o$ contains the $\tau$ users with the lowest costs. Hence, every user with a cost lower than $\hat{p}$ must be in $P_o$ since $\hat{p}$ is in $P_o$. The lemma now follows by observing that the definition of $\hat{P}$ implies $c(p) < c(\hat{p})$ for every user $p \in \hat{P}$. 
\end{proof}

\begin{corollary}\label{cor:middle_value}
There exists a value $\ell(P, B)$ independent of the random coins of {\TPM} such that:
\begin{compactitem}
	\item $c(p) \leq \ell(P, B) \leq v(b)$ for every user $p \in P_o$ and slot $b \in B_o$
	\item Whenever the event \hyperlink{event:E_prime}{$\cE'$} occurs, $c(p) \leq \ell(P, B) \leq v(b)$ for every user $p \in \hat{P}$ and slot $b \in \hat{B}$.
\end{compactitem}
\end{corollary}
\begin{proof}
Let $\ell(P, B)$ be the value of the slot at location $\tau$ of the canonical assignment $S_c(P, B)$. Clearly, $\ell(P, B)$ is independent of the random coins of {\TPM}, as required. Additionally, for every slot $b \in B_o$ it holds that $v(b) \geq \ell(P, B)$ since $b$ must be located at some location of $S_c(P, B)$ between $1$ and $\tau$. On the other hand, let $p_\tau$ be the user at location $\tau$ of $S_c(P, B)$. Since the size of $S_c(P, B)$ is $\tau$, $p_\tau$ must be assigned to the slot at location $\tau$ of $S_c(P, B)$, which implies $c(p_{\tau}) \leq \ell(P, B)$.\footnote{In fact, we even have $c(p_{\tau}) < \ell(P, B)$ since the tie-breaking rule defined in Section~\ref{sec:comparison} guarantees that the value of a slot is never equal to the cost of a user.} Moreover, for every user $p \in P_o$ it holds that $c(p) \leq c(p_\tau) \leq \ell(P, B)$ since $p$ must be located at some location of $S_c(P, B)$ between $1$ and $\tau$.

The lemma now follows since Lemma~\ref{lem:upper_inclusion} shows that the event $\cE'$ implies that every user $p \in \hat{P}$ belongs also to $P_o$, and every slot $b \in \hat{B}$ belongs also to $B_o$.
\end{proof}

\begin{lemma} \label{lem:lower_inclusion}
The event \hyperlink{event:E_prime}{$\cE'$} implies $\tilde{P} \cap P(M_1) \subseteq \hat{P}$ and $\tilde{B} \cap B(A_1) \subseteq \hat{B}$.
\end{lemma}
\begin{proof}
We prove the first inclusion. The other inclusion is analogous. The claim about $\tilde{P} \cap P(M_1)$ is trivial when $\tilde{P}$ is empty. Thus, we can assume throughout the proof that $\tilde{P}$ is non-empty. Observation~\ref{obs:length_characterization} and the definition of $\cE'$ imply:
\begin{align*}
	|S_c(P(M_2), B(A_2))|
	\geq{} &
	\min \{|P_o \cap P(M_2)|, |B_o \cap B(A_2)|\}\\
	\geq{} &
	\min\{|P_o|/2 - \sqrt[3]{\alpha} \cdot \tau, |B_o|/2 - \sqrt[3]{\alpha} \cdot \tau\}
	=
	(1/2 - \sqrt[3]{\alpha})\tau
	\enspace.
\end{align*}
Recall that $\alpha \geq \tau^{-1}$, and thus, $\sqrt[3]{\alpha} \cdot \tau \geq 1$. Using this inequality and the definition of $\cE'$ again gives:
\begin{align*}
	|\tilde{P} \cap P(M_2)|
	\leq{} &
	|\tilde{P}|/2 + \sqrt[3]{\alpha} \cdot \tau
	=
	\lceil (1 - 11\sqrt[3]{\alpha})\tau \rceil/2 + \sqrt[3]{\alpha} \cdot \tau
	\leq
	(1/2 - 4\sqrt[3]{\alpha})\tau - 1 + \sqrt[3]{\alpha} \cdot \tau\\
	={} &
	(1/2 - 3\sqrt[3]{\alpha})\tau - 1
	\leq
	(1 - 4\sqrt[3]{\alpha}) \cdot (1/2 - \sqrt[3]{\alpha})\tau - 1\\
	\leq{} &
	(1 - 4\sqrt[3]{\alpha}) \cdot |S_c(P(M_2), B(A_2))| - 1
	\enspace,
\end{align*}
which implies, since $|\tilde{P} \cap P(M_2)|$ is integral,
\begin{equation} \label{eq:threshold_above_tilde}
	|\tilde{P} \cap P(M_2)|
	\leq
	\lceil (1 - 4\sqrt[3]{\alpha}) \cdot |S_c(P(M_2), B(A_2))| \rceil - 1
	<
	\lceil (1 - 4\sqrt[3]{\alpha}) \cdot |S_c(P(M_2), B(A_2))| \rceil
	\enspace.
\end{equation}

Inequality~\eqref{eq:threshold_above_tilde} and the observation that the users of $\tilde{P} \cap P(M_2)$ are the users with the lowest costs in $P(M_2)$ imply together that $\hat{p}$ is a user of $P(M_2)$ which does not belong to $\tilde{P} \cap P(M_2)$, and therefore, does not belong to $\tilde{P}$ either. 
On the other hand, $\tilde{P}$ contains the $\lceil (1 - 11\sqrt[3]{\alpha})\tau \rceil$ users with the lowest costs. Hence, every user $p \in \tilde{P}$ has a cost smaller than $c(\hat{p})$ since $\hat{p}$ does not belong to $\tilde{P}$. The lemma now follows by observing that the definition of $\hat{P}$ implies that $p \in \hat{P}$ for every user $p \in P(M_1)$ obeying $c(p) < c(\hat{p})$. 
\end{proof}

\begin{corollary} \label{cor:size_bounds}
The event \hyperlink{event:E_prime}{$\cE'$} implies: $-6.5\sqrt[3]{\alpha} \cdot \tau \leq |\hat{P}| - \tau/2 \leq \sqrt[3]{\alpha} \cdot \tau$ and $-6.5\sqrt[3]{\alpha} \cdot \tau \leq |\hat{B}| - \tau/2 \leq \sqrt[3]{\alpha} \cdot \tau$.
\end{corollary}
\begin{proof}
We prove here only the bounds on the size of $\hat{P}$. The bounds on the size of $\hat{B}$ are analogous. By Lemma~\ref{lem:upper_inclusion}, $\hat{P} \subseteq P_o$. On the other hand, by definition, $\hat{P} \subseteq P(M_1)$. Thus, we get: $\hat{P} \subseteq P_o \cap P(M_1)$. Combining this inclusion with the definition of $\cE'$ gives:
\begin{align*}
	|\hat{P}|
	\leq{} &
	|P_o \cap P(M_1)|
	=
	|P_o| - |P_o \cap P(M_2)|\\
	\leq{} &
	|P_o| - [|P_o|/2 - \sqrt[3]{\alpha} \cdot \tau]
	=
	|P_o|/2 + \sqrt[3]{\alpha} \cdot \tau
	=
	\tau/2 + \sqrt[3]{\alpha} \cdot \tau
	\enspace.
\end{align*}
On the other hand, by Lemma~\ref{lem:lower_inclusion} and the definition of $\cE'$,
\begin{align*}
	|\hat{P}|
	\geq{} &
	|\tilde{P} \cap P(M_1)|
	=
	|\tilde{P}| - |\tilde{P} \cap P(M_2)|
	\geq
	|\tilde{P}| - [|\tilde{P}|/2 + \sqrt[3]{\alpha} \cdot \tau]\\
	={} &
	|\tilde{P}|/2 - \sqrt[3]{\alpha} \cdot \tau
	\geq
	\lceil (1 - 11\sqrt[3]{\alpha})\tau \rceil / 2 - \sqrt[3]{\alpha} \cdot \tau
	\geq
	\tau/2 - 6.5\sqrt[3]{\alpha} \cdot \tau
	\enspace.
	\qedhere
\end{align*}
\end{proof}

We can now define the event $\cE_1$ referred to by Lemma~\ref{lem:event_summary}. The event $\cE_1$ is the event that $\cE'$ happens and in addition the following two inequalities also hold:
\begin{center}
\begin{tabular}{llll}
	(i)  & $|\hat{B} \setminus B(A_L)| \leq |\hat{P}|$ \hspace{1cm} & (ii) & $|\hat{P} \setminus P(M_L)| \leq |\hat{B}|$
\end{tabular}
\end{center}

Next, let us prove Lemma~\ref{lem:event_summary}. For ease of the reading, we first repeat the lemma itself.

\begin{replemma}{lem:event_summary}
There exists an event $\cE_1$ of probability at least $1 - 10e^{-2/\sqrt[3]{\alpha}}$ such that $\cE_1$ implies the following:
\begin{center}
\begin{tabular}{llll}
	(i)  & $\tilde{B} \cap B(A_1) \subseteq \hat{B}$ \hspace{3cm} & (iii) & $|\hat{P} \setminus P(M_L)| \leq |\hat{B}|$\\
	(ii) & $\tilde{P} \cap P(M_1) \subseteq \hat{P}$ \hspace{3cm} & (iv)  & $|\hat{B} \setminus B(A_L)| \leq |\hat{P}|$ \\
	(v)  & \multicolumn{3}{p{11cm}}{$c(p) \leq \ell(P, B) \leq v(b)$ for every user $p \in \hat{P}$ and slot $b \in \hat{B}$, where $\ell(P, B)$ is a value which is independent of the random coins of {\TPM} and obeys $c(p) \leq \ell(P, B) \leq v(b)$ for every $p \in P_o$ and $b \in B_o$.}
\end{tabular}
\end{center}
\end{replemma}
\begin{proof}
By definition, the event $\cE_1$ implies the inequalities: $|\hat{B} \setminus B(A_L)| \leq |\hat{P}|$ and $|\hat{P} \setminus P(M_L)| \leq |\hat{B}|$. Additionally, $\cE_1$ implies the event $\cE'$, which, by Corollary~\ref{cor:middle_value} and Lemma~\ref{lem:lower_inclusion}, implies the other things that should follow from $\cE_1$ by the lemma. Hence, the only thing left to prove is that the probability of $\cE_1$ is at least $1 - 10e^{-2/\sqrt[3]{\alpha}}$.

If $17\sqrt[3]{\alpha} \geq 1$, then $A_L = A$ and $M_L = M$, which implies that the two inequalities $|\hat{B} \setminus B(A_L)| \leq |\hat{P}|$ and $|\hat{P} \setminus P(M_L)| \leq |\hat{B}|$ are trivial. Hence, the events $\cE_1$ and $\cE'$ are equivalent in this case, and thus, the probability of $\cE_1$ is at least $1 - 8e^{-2/\sqrt[3]{\alpha}}$ by Observation~\ref{obs:E_prime_probability}. Therefore, it is safe to assume in the rest of the proof that $17\sqrt[3]{\alpha} < 1$.

Our plan is to prove the inequality $\Pr[\cE_1 \mid \cE'] \geq 1 - 2e^{-2/\sqrt[3]{\alpha}}$. Notice that this inequality indeed implies the lemma since it implies:
\[
	\Pr[\cE_1]
	=
	\Pr[\cE'] \cdot \Pr[\cE_1 \mid \cE']
	\geq
	(1 - 8e^{-2/\sqrt[3]{\alpha}}) \cdot (1 - 2e^{-2/\sqrt[3]{\alpha}})
	\geq
	1 - 10e^{-2/\sqrt[3]{\alpha}}
	\enspace.
\]

The event $\cE'$ is fully determined by way {\TPM} partitions $M$ and $A$ into $M_1, M_2, A_1$ and $A_2$. Thus, it is enough to show that for every fixed partition for which the event $\cE'$ holds, the event $\cE_1$ holds with probability at least $1 - 2e^{-2/\sqrt[3]{\alpha}}$. Notice that the sets $\hat{P}$ and $\hat{B}$ become deterministic once we fix the partition of $A$ and $M$. Hence, either $|\hat{B}| \leq |\hat{P}|$, which implies that the inequality $|\hat{B} \setminus B(A_L)| \leq |\hat{P}|$ holds regardless of the choice of $A_L$, or $|\hat{P}| \leq |\hat{B}|$, which implies that the inequality $|\hat{P} \setminus P(M_L)| \leq |\hat{B}|$ holds regardless of the choice of $M_L$. In both cases, all we need to show is that the other inequality holds with probability at least $1 - 2e^{-2/\sqrt[3]{\alpha}}$ over the random choice of $A_L$ and $M_L$.

Let us assume, without loss of generality, that $|\hat{B}| \leq |\hat{P}|$. By the above discussion, all we need to prove is that $\Pr[|\hat{P} \setminus P(M_L)| \leq |\hat{B}|] \geq 1 - 2e^{-2/\sqrt[3]{\alpha}}$, where the probability is over the random choice of $M_L$. By Corollary~\ref{cor:size_bounds}:
\begin{align*}
	\Pr[|\hat{P} \setminus P(M_L)| > |\hat{B}|]
	\leq{} &
	\Pr[|\hat{P} \setminus P(M_L)| > \tau/2 - 6.5\sqrt[3]{\alpha} \cdot \tau]\\
	\leq{} &
	\Pr[|\hat{P} \setminus P(M_L)| > (1 - 17\sqrt[3]{\alpha})(\tau/2 + \sqrt[3]{\alpha} \cdot \tau) + \sqrt[3]{\alpha} \cdot \tau]\\
	\leq{} &
	\Pr[|\hat{P} \setminus P(M_L)| > (1 - 17\sqrt[3]{\alpha}) \cdot |\hat{P}| + \sqrt[3]{\alpha} \cdot \tau]
	\enspace.
\end{align*}
Notice now that $\hat{P} \setminus P(M_L)$ has the same distribution as $\hat{P}[1 - 17\sqrt[3]{\alpha}]$. Hence, by Lemma~\ref{lem:length_concentration}:
\begin{align*}
	\Pr[|\hat{P} \setminus P(M_L)| > |\hat{B}|]
	\leq{} &
	\Pr[|\hat{P} \setminus P(M_L)| > (1 - 17\sqrt[3]{\alpha}) \cdot |\hat{P}| + \sqrt[3]{\alpha} \cdot \tau]\\
	\leq{} &
	\Pr[||\hat{P}[1 - 17\sqrt[3]{\alpha}]| - (1 - 17\sqrt[3]{\alpha}) \cdot |\hat{P}|| > \sqrt[3]{\alpha} \cdot \tau]
	\leq
	2e^{-2/\sqrt[3]{\alpha}}
	\enspace.
	\qedhere
\end{align*}
\end{proof}

%% file: SecConclusions.tex
\section{Conclusions}

We considered mechanisms for double-sided markets that interact with strategic players where at least one side of the market has players with multi-dimensional strategic spaces. In particular, we explored the question of how many sides of the market can have players with multi-dimensional strategic spaces, while still allowing for mechanisms that both maintain the desired basic economic properties (individual rationality, incentive compatibility and budget balance) and suffer only a (small) bounded loss compared to the socially optimal outcome. To answer this question we presented two mechanisms. One mechanism that is deterministic and allows one side to have players with multi-dimensional strategic spaces, and another mechanism that is randomized and allows two sides to have players with multi-dimensional strategic spaces.

Our mechanisms significantly extend the literature on
%mediated mechanisms as they introduce monetary and quasi linear utilities for mediators with a captive audience. Furthermore, our mechanisms significantly extend the literature on
trade reduction---a technique used to achieve the basic economic properties of individual rationality, incentive compatibility and budget balance in a multi-sided market. While all the previous trade reduction solutions dealt with players having single dimensional strategic spaces, our deterministic algorithm performs a non-binary trade reduction which leads to the first trade reduction solution applying to multi-dimensional strategic spaces. As given, this deterministic trade reduction solution deals with two sides: advertisers and mediators; however, we believe that our ideas are more general, and we currently work on extending the mechanism in order to handle more market sides.

From a more practical point of view, our study is motivated by a foreseeable future form of online advertising in which users are incentivized to share their information via participation in a mediated online advertising exchange. Our model captures some of the challenges introduced by such exchanges, but does not encompass their online nature. In this sense our results in this paper can be seen as a step towards the study of a more realistic model which also captures the online nature of ad exchanges.%\footnote{In the follow up paper of~\cite{FG16a} we present one possible way to introduce online aspects into our model.}

%% file: Offline_Algorithms.bbl
\begin{thebibliography}{10}

\bibitem{AMT09}
Itai Ashlagi, Dov Monderer, and Moshe Tennenholtz.
\newblock Mediators in position auctions.
\newblock {\em Games and Economic Behavior}, 67:2--21, 2009.

\bibitem{BW03}
M.~Babaioff and W.~E. Walsh.
\newblock Incentive-compatible, budget-balanced, yet highly efficient auctions
  for supply chain formation.
\newblock In {\em EC}, 2003.

\bibitem{BNP09}
Moshe Babaioff, Noam Nisan, and Elan Pavlov.
\newblock Mechanisms for a spatially distributed market.
\newblock {\em Games and Economic Behavior}, 66(2):660--684, 2009.

\bibitem{BEKT07}
L.~Blume, D.~Easley, J.~Kleinberg, and E.~Tardos.
\newblock Trading networks with price-setting agents.
\newblock In {\em EC}, pages 13--16, 2007.

\bibitem{CM06}
B.~Chaib-draa and J.~P. Muller.
\newblock Multiagent-based supply chain management.
\newblock {\em Berlin: Springer Verlang}, 2006.

\bibitem{C71}
E.~H. Clarke.
\newblock Multipart pricing of public goods.
\newblock {\em Public Choice}, 2:17--33, 1971.

\bibitem{FMMP10}
Jon Feldman, Vahab~S. Mirrokni, S.~Muthukrishnan, and Mallesh~M. Pai.
\newblock Auctions with intermediaries: extended abstract.
\newblock In {\em EC}, pages 23--32, 2010.

\bibitem{GGP07}
Mira Gonen, Rica Gonen, and Elan Pavlov.
\newblock Generalized trade reduction mechanisms.
\newblock In {\em EC}, pages 20--29, 2007.

\bibitem{G73}
T.~Groves.
\newblock Incentives in teams.
\newblock {\em Econometrica}, 41:617--631, 1973.

\bibitem{M92}
R.~P. McAfee.
\newblock dominant strategy double auction.
\newblock {\em Journal of Economic Theory}, 56:434--450, 1992.

\bibitem{MS83}
R.~B. Myerson and M.~A. Satterthwaite.
\newblock Efficient mechanisms for bilateral trading.
\newblock {\em Journal of Economic Theory}, 29:265--281, 1983.

\bibitem{RCJZ01}
R.~Roundy, R.~Chen, G.~Janakriraman, and R.~Q. Zhang.
\newblock Efficient auction mechanisms for supply chain procurement.
\newblock {\em School of Operations Research and Industrial Engineering,
  Cornell University}, 2001.

\bibitem{HHA16}
Erel Segal{-}Halevi, Avinatan Hassidim, and Yonatan Aumann.
\newblock {SBBA:} {A} strongly-budget-balanced double-auction mechanism.
\newblock In {\em SAGT}, pages 260--272, 2016.

\bibitem{SGP14}
Lampros~C. Stavrogiannis, Enrico~H. Gerding, and Maria Polukarov.
\newblock Auction mechanisms for demand-side intermediaries in online
  advertising exchanges.
\newblock In {\em AAMAS}, pages 5--9, 2014.

\bibitem{V61}
W.~Vickrey.
\newblock Counterspeculation, auctions and competitive sealed tenders.
\newblock {\em Journal of Finance}, 16:8--37, 1961.

\end{thebibliography}
